\documentclass[cleveref,runningheads]{llncs}
\usepackage[usenames,dvipsnames]{xcolor}
\usepackage{enumerate}
\usepackage{mathtools} 
\usepackage{amsmath}
\usepackage{latexsym, amsfonts, stmaryrd, amssymb, graphicx, wrapfig}
\usepackage{amsbsy} 
\usepackage{marvosym} 
\usepackage{eucal}
\usepackage{proof}
\usepackage{tikz}
\usetikzlibrary{arrows,arrows.meta,shapes,snakes,automata,backgrounds,petri,
decorations.pathmorphing,positioning,fit}
\usepackage{dsfont} 
\usepackage{xifthen} 
\usepackage{cases} 
\usepackage{environ, thmtools, thm-restate}  
\usepackage{nicefrac}
\usepackage{DotArrow}
\usepackage{keyval}
\usepackage{todonotes}
\usepackage{hyperref}
\usepackage{cleveref}
\usepackage{soul}

\usepackage[nomargin,inline,index,draft]{fixme} 
\fxusetheme{color}
\usepackage{ifdraft}                
\FXRegisterAuthor{michele}{anmichele}{\color{blue} {\underline{michele}}}
\FXRegisterAuthor{claudio}{anclaudio}{\color{red} {\underline{claudio}}}
\FXRegisterAuthor{irek}{anirek}{\color{orange} {\underline{irek}}}
\FXRegisterAuthor{added}{anadded}{\color{red} {\underline{added}}}

\usepackage{appendix}
\usepackage[inline,shortlabels]{enumitem} 

\newcommand{\CF}[1]{\ensuremath{\mathsf{CF}(#1)}}

\newcommand{\trans}[1]{\ensuremath{\,[\/{#1}\/\rangle}\,}
\newcommand{\pre}[1]{\ensuremath{\!~^{\bullet}{#1}}}
\newcommand{\post}[1]{\ensuremath{{#1} {^{\bullet}}}}

\newcommand{\hist}[1]{\ensuremath{\lfloor #1 \rfloor}}
\newcommand{\future}[1]{\ensuremath{\lceil #1 \rceil}}

\newcommand{\nat}{\ensuremath{\mathbb{N}}}

\newcommand{\flt}[1]{\ensuremath{[\![{#1}]\!]}}
\newcommand{\pes}{\textsc{pes}}

\newcommand{\rcn}{\textsc{rcn}}
\newcommand{\cn}{\textsc{on}}

\newcommand{\setenum}[1]{\{#1\}}
\newcommand{\setcomp}[2]{\{{#1} \mid {#2}\}}

\newcommand{\reachMark}[1]{\ensuremath{\mathcal{M}_{#1}}}

\newcommand{\firseq}[2]{\ensuremath{\mathcal{R}^{#1}_{#2}}}
\newcommand{\states}[1]{\ensuremath{\mathsf{St}(#1)}}

\newcommand{\lead}[1]{\ensuremath{\mathit{lead}(#1)}}
\newcommand{\start}[1]{\ensuremath{\mathit{start}(#1)}}
\newcommand{\fs}{\textsf{fs}}

\newcommand{\remains}[1]{\ensuremath{\mathit{tail}(#1)}}
\newcommand{\MC}[1]{\ensuremath{\sim}}

\newcommand{\Conf}[2]{\ensuremath{\mathsf{Conf}_{#2}(#1)}}

\newcommand{\Comment}[1]{}
\newcommand{\un}[1]{\underline{#1}}
\newcommand{\pr}{\mathrel{\triangleright}}

\newcommand{\changed}[2]{#2}

\newcommand{\fe}{\mathtt{f}}
\newcommand{\re}{\mathtt{r}}

\newcommand{\anR}{U}
\newcommand{\anr}{u}

\newcommand{\Er}{U}
\newcommand{\er}{u}

\begin{document}

\title{Reversible Causal Nets and Reversible Event Structures}
\titlerunning{Reversible Causal Nets and Reversible Event Structures}

\author{Hern\'an Melgratti\inst{1} \and
Claudio Antares Mezzina\inst{2} \and
Iain Phillips\inst{3} \and
G. Michele Pinna\inst{4} \and
Irek Ulidowski\inst{5} 
}
\authorrunning{H. Melgratti, C. A. Mezzina, I. Phillips, G. M. Pinna, I. Ulidowski}

\institute{ICC - Universidad de Buenos Aires - Conicet, Argentina \and
Dipartimento di Scienze Pure e Applicate, Universit\`a di Urbino, Italy \and
 Imperial College London, England, UK \and
 Universit\`a di Cagliari, Italy \and 
 University of Leicester, England, UK}
\maketitle

\begin{abstract}

One of the well-known results in concurrency theory concerns the relationship between 
event structures and occurrence nets: 
an occurrence net can be associated with
a prime 
event structure, and vice versa. More generally, the relationships between various 
forms of event structures and suitable forms of nets have been long established. 
Good examples  are the close relationship between inhibitor event structures and 
inhibitor occurrence nets, or between asymmetric event structures and asymmetric occurrence nets. 
Several forms of event structures suited for the modelling of reversible 
computation have recently been developed; also a method for reversing occurrence nets has 
been proposed. This paper bridges the gap between reversible event structures 
and reversible nets. We introduce the notion of reversible causal net, which is a 
generalisation of the notion of reversible unfolding. We show that reversible 
causal nets correspond precisely to a subclass of reversible prime event structures, the
causal reversible prime event structures.

\end{abstract}


\section{Introduction}
\label{sec:intro}
Event structures and nets are closely related. Since the seminal papers
by Nielsen, Plotkin and Winskel \cite{NPW:PNES} and Winskel \cite{Win:ES}, the relationship among
nets and event structures has been considered as a pivotal characteristic of
concurrent systems, as testified by numerous papers in the literature addressing
 event structures and nets. 
The ingredients of an event structure are, beside a set of events, a number of relations that are
used to express which events can be part of a configuration (the snapshot of a concurrent 
system), modelling a consistency predicate, and how events can be added to reach another 
configuration, modelling the dependencies among the (sets of) events.
On the net side, the ingredients boil down to constraints on how transitions may be executed,
and usually have a structural flavour.

Since the introduction of event structures there has been a flourish of
investigations into the possible relations among events, giving rise to a number
of different definitions of event structures.
%
We recall some of them, without the claim of completeness. 
First to mention are the classical \emph{prime} event structures 
\cite{Win:ES} where
the dependency between events, called \emph{causality}, is given by a partial order and the consistency
is determined by a \emph{conflict} relation.
\emph{Flow} event structures \cite{Bou:FESFN} drop the requirement that the dependency should
be a partial order, and \emph{bundle} event structures \cite{Langerak:1992:BES} are able to represent 
OR-causality by allowing each event to be caused by a member of a bundle of events.
\emph{Asymmetric} event structures \cite{BCM:CNAED} introduce the notion of weak causality that can 
model asymmetric conflicts. \emph{Inhibitor} event structures \cite{BBCP:rivista} are able
to faithfully capture the dependencies among events which arise in the presence of read and inhibitor arcs. 
%
In \cite{BCP:LenNets} event structures where the causality relation may be circular are
investigated, and in \cite{AKPN:lmcs18} the notion of dynamic causality is considered.
Finally, we mention the quite general approach presented in \cite{GP:CSESPN}, where there is
a unique relation, akin to a \emph{deduction relation}. 
To each of the aforementioned event structures 
a particular class of nets corresponds. To prime event structures we have 
\emph{occurrence nets}, to flow event structures we have flow nets, to bundle event structures
we have \emph{unravel nets} \cite{CaPi:PN14}, to asymmetric and inhibitor event structures we
have \emph{contextual nets} \cite{BCM:CNAED,BBCP:rivista}, to event structures with circular
causality we have \emph{lending nets} \cite{BCP:LenNets}, to those with dynamic causality we have
\emph{inhibitor unravel nets} \cite{CP:soap17} and finally to the ones presented in 
\cite{GP:CSESPN}  \emph{1-occurrence nets} are associated.

Recently a new type of event structure tailored to model \emph{reversible} computation has been
introduced and studied \cite{PU:jlamp15,UPY:RES-NGC18}. In particular, in \cite{PU:jlamp15},
\emph{reversible prime event structures} have been introduced. In this kind of event structure 
two relations are added: 
the \emph{reverse causality} relation and the \emph{prevention} relation. The first 
one is a standard dependency relation: in order to reverse an event some other events must be
present. The second relation, on the contrary, identifies those events whose presence \emph{prevents}
the event being reversed. 
This kind of event structure is able to model different
flavours of reversible computation 
such as causal-consistent reversibility~\cite{rccs,ccsk,rhotcs} and out-of-causal-order 
reversibility~\cite{PhiUliYuen12,UK16}. 
Causally consistent reversibility relates reversibility with causality: an event can be undone provided that all of its effects have been undone. This allows the system to get back to a past state, which was possible to reach by just the normal (forward) computation.
This notion of reversibility is natural in reliable distributed systems since when an error occurs the system tries to go back to a past consistent state. Examples of how causal consistent reversibility is used to model reliable systems are  
transactions~\cite{DanosK05,LaneseLMSS13} and rollback protocols~\cite{VassorS18}.
Also, causally consistent reversibility can be used for program analysis and  debugging \cite{GiachinoLM14,LanesePV19}. 
 On the other hand, out-of-causal-order reversibility  does not preserve causes, and it is suitable to model biochemical reaction where, for example, a bond can be undone leading to a different state which was not present before.

Reversibility in Petri nets has been studied in \cite{PhilippouP18,MMU:coordination19} with two different approaches. 
In \cite{PhilippouP18} reversibility for acyclic Petri net is solved by
relying on a new kind of tokens, called bonds, that
keep track of the execution history. Bonds are rich enough for allowing other approaches to
reversibility such as out-of-causal order and causal consistent reversibility. In
 \cite{MMU:coordination19}
a notion of \emph{unfolding} of a P/T (place/transition) net, where all the transitions can be reversed, has been
proposed. In particular,  by resorting to standard notions of the Petri net  theory \cite{MMU:coordination19} provides a causally-consistent reversible semantics 
for P/T nets. This exploits the well-known unfolding of P/T nets into occurrence nets 
\cite{Win:ES}, and is done by adding for each transition its reversible counterpart.

In this paper we start 
our research quest towards studying what kind of nets can be
associated to reversible prime event structures. To this aim we introduce the notion of
a \emph{reversible causal net} which is an occurrence net enriched with 
transitions which operationally
\emph{reverse} the effects of executing some others.

We associate to each reversing transition (event in the occurrence net dialect) a unique transition, 
which is in charge of producing
 the effects that the reversing transition
has to undo. 
To execute a reversing event in a reversible causal net 
the events caused by the event to be reversed (the one associated to the reversing one)
have to be reversed as well, and if this is not possible then the reversing event cannot be
executed. This corresponds, in the reversible event structure, to the fact that such 
events \emph{prevent} the reversing event from happening. A reversible causal 
net where the reversing events 
have been removed is just an occurrence net.
This discussion suggests the easiest way of relating
reversible causal nets and reversible prime event structure: the causal relation is
the one induced by the occurrence net, and the prevention relation is the one induced
by the events that are in the future of the one to be reversed. The reversible causality relation
is the basic one: in order to reverse an event the event itself must be present. 
What is obtained from a reversible causal net is a \emph{causal} reversible prime event structure,
which is a subclass of reversible prime event structures.  

When we start from a causal reversible prime event structure, it is possible to obtain
a reversible causal net. The ingredients that are used are just the causal relation and
the set of reversible events. 
Thus, if we start from a causal
reversible prime event structure, the obtained reversible causal net  
has the same set of configurations.
Hence, the precise correspondence is between causal reversible prime event structures and 
reversible causal nets. 
%
%
This relation is made clear also by turning causal reversible prime event structures and
reversible causal nets into categories. Then the constructions associating reversible causal nets to 
causal reversible prime event structures can be turned into functors and these functors form a coreflection.
This implies that the notion of reversible causal net is the appropriate one when dealing with causal
prime event structures.

\paragraph{Structure of the paper.}
Section \ref{sec:preliminaries} reviews some preliminary notions for nets and
event structures, including reversible prime event structures. Section \ref{sec:cn-and-pes} recalls
the well-known relationship between prime event structures and occurrence nets. The core of the
paper is Section \ref{sec:rcn-and-rpes} where we first introduce reversible causal nets and then
we show how to obtain a reversible causal net from an occurrence net. We then 
show how to associate a causal reversible prime event structure to a reversible causal net, 
and vice versa. 
We sum up our findings in Section \ref{sec:category} where we introduce a notion of morphism for 
reversible causal nets and show that this gives a category which is related to the subcategory of causal
reversible prime event structures.
Section \ref{sec:conc} concludes the paper.

\section{Preliminaries}\label{sec:preliminaries}
 We denote with  $\nat$ the set of natural numbers.
 Let $A$ be a set, a {\em multiset\/} of $A$ is a function $m : A
 \rightarrow \nat$.
 The set of multisets of $A$ is denoted by $\mu A$.  
 %
We assume the usual operations on multisets such as union $+$ and difference $-$.

 We write $m \subseteq m'$ if $m(a) \leq m'(a)$ for all $a \in A$.  
 For $m\in \mu A$, we denote with $\flt{m}$ the multiset defined as $\flt{m}(a)
 = 1$ if $m(a) > 0$ and $\flt{m}(a) = 0$ otherwise. 
 When a multiset $m$ of $A$ is a set, 
 \emph{i.e.} $m = \flt{m}$, we write
 $a\in m$ to denote that $m(a) \neq 0$, and often confuse the
 multiset $m$ with the set $\setcomp{a\in A}{m(a) \neq 0}$. 
 Furthermore we use the standard set operations like $\cap$, $\cup$ or
 $\setminus$.
 
 Given a set $A$ and a relation $<\ \subseteq A\times A$, we say that $<$ is 
 an \emph{irreflexive} partial order whenever
it is irreflexive and transitive. 
We shall write $\leq$ for the reflexive closure of a  partial order $<$.
 
\subsection{Petri nets} 
We review the notion of Petri net along with some auxiliary notions.
\begin{definition}
   A \emph{Petri net} is a 4-tuple 
   $N = \langle S, T, F, \mathsf{m}\rangle$, where
   $S$ is a set of {\em places} and $T$ is a set of {\em transitions} 
   (with $S \cap T = \emptyset$), 
   {$F \subseteq (S\times T)\cup (T\times S)$}
    is the 
    \emph{flow} relation, and
   $\mathsf{m}\in \mu S$ is called the {\em initial marking}.
 \end{definition}
 Petri nets are depicted as usual. 
 Given a net $N = \langle S, T, F, \mathsf{m}\rangle$ and  $x\in S\cup T$, we define the following
 multisets:  
 {$\pre{x} = \{y\ |\ (y,x)\in F\}$}
  and 
 {$\post{x} = \{y\ |\ (x,y)\in F\}$}.
 If $x\in S$ then 
 $\pre{x} \in \mu T$ and  $\post{x} \in \mu T$;
 analogously, 
 if $x\in T$ then $\pre{x}\in\mu S$ and $\post{x} \in \mu S$.
 A multiset of transitions $A\in \mu T$, called \emph{step}, 
 is enabled at a marking $m\in \mu S$, denoted by $m\trans{A}$,
 whenever $\pre{A} \subseteq m$, where $\pre{A} = \sum_{x\in\flt{A}}\ A(x)\cdot\pre{x}$. 
 A step $A$ enabled at a marking $m$ can \emph{fire} and its firing produces 
 the marking $m' = m - \pre{A} + \post{A}$, where 
 $\post{A} = \sum_{x\in\flt{A}}\ A(x)\cdot\post{x}$.
 The firing of $A$ at a marking $m$ is denoted by $m\trans{A}m'$.
 We assume that each transition $t$ of a net $N$ is such that $\pre{t}\neq\emptyset$,
 meaning that no transition may fire \emph{spontaneously}.
 Given a generic marking $m$ (not necessarily  the initial one), 
 the (step) \emph{firing sequence} ({shortened as} \fs) of  
 $N = \langle S, T, F, \mathsf{m}\rangle$  starting at $m$ is
 defined as: 
   ($i$)  $m$ is a firing sequence (of length 0), and 
  ($ii$) if $m\trans{A_1}m_1$ $\cdots$ $m_{n-1}\trans{A_n}m_n$ is a firing sequence 
        and $m_n\trans{A}m'$,  
        then also $m\trans{A_1}m_1$ $\cdots$ $m_{n-1}\trans{A_n}m_n\trans{A}m'$
        is a firing sequence.
 %
 Let us note that each step $A$ such that $|A| = n$ 
 can be written as $A_1 + \cdots + A_n$ where for each $1 \leq i\leq n$ it holds that 
 $A_i = \flt{A_i}$ and  $|A_i| = 1$, and $m\trans{A}m'$ iff for each decomposition
 of $A$ in $A_1 + \cdots + A_n$, we have that $m\trans{A_1}m_1\dots m_{n-1}\trans{A_n}m_n = m'$.
 When $A$ is a singleton, \emph{i.e.} 
 $A = \setenum{t}$, we write 
 $m\trans{t}m'$.   
 The set of firing sequences of a net $N$ 
 starting at a marking $m$ is denoted by $\firseq{N}{m}$ and it is ranged over by $\sigma$.
 Given a \fs\ $\sigma = m\trans{A_1}\sigma'\trans{A_n}m_n$, we denote with
 $\start{\sigma}$  the marking $m$, with $\lead{\sigma}$ the marking $m_n$ 
 and with $\remains{\sigma}$ the 
 \fs\ $\sigma'\trans{A_n}m_n$. 
 Given a net $N = \langle S, T, F, \mathsf{m}\rangle$, a marking $m$ is \emph{reachable} 
 iff there exists a 
 \fs\ $\sigma \in \firseq{N}{\mathsf{m}}$ 
 such that $\lead{\sigma}$ is $m$. 
 The set of reachable markings of $N$ is
 $\reachMark{N} = \bigcup_{\sigma\in\firseq{N}{\mathsf{m}}} \lead{\sigma}$.
 Given a \fs\ $\sigma = m\trans{A_1}m_1\cdots m_{n-1}\trans{A_n}m'$, 
%
 we write  
 $X_{\sigma} = \sum_{i=1}^{n} A_i$
 for the multiset of transitions associated to  \fs. 
%
%
We call $X_{\sigma}$ a \emph{state} of the net and write
 \(
   \states{N} = \setcomp{X_{\sigma}\in \mu T}{\sigma\in\firseq{N}{\mathsf{m}}}
 \)
 for the set of states of $N$.


\begin{definition}  
 A net $N = \langle S, T, F, \mathsf{m}\rangle$ is said to be \emph{safe} if
 each marking $m\in \reachMark{N}$ is such that $m = \flt{m}$. 
\end{definition}

%
%
In this paper we consider safe nets $N =  \langle S, T, F, \mathsf{m}\rangle$ 
where each transition can be fired, \emph{i.e.}
$\forall t\in T\ \exists m\in \reachMark{N}.\ m\trans{t}$, and each place is marked in a 
computation, \emph{i.e.} $\forall s\in S\ \exists m\in \reachMark{N}.\ m(s) = 1$.

\subsection{Prime event structures}

We now recall the notion of prime event structure~\cite{Win:ES}.

\begin{definition}\label{de:pes-winskel}
 A \emph{prime event structure ({\pes})} is a triple $P = (E, <, \#)$, where 
 \begin{itemize}
  \item $E$ is a  countable set of \emph{events},
  \item $<\ \subseteq E\times E$ is
an irreflexive partial order called the \emph{causality relation}, such that
        $\forall e\in E$. $\setcomp{e'\in E}{e' < e}$ is finite, and 
  \item $\#\ \subseteq E\times E$ is a \emph{conflict relation}, which is irreflexive,
    symmetric and \emph{hereditary} relation with respect to $<$: if  $e\ \#\ e' < e''$, then
    $e\ \#\ e''$ for all $e,e',e'' \in E$.
 \end{itemize}       
\end{definition}

Given an event $e\in E$,  $\hist{e}$  denotes the set $\setcomp{e'\in E}{e'\leq e}$. 
{A subset of events $X \subseteq E$ is left-closed if $\forall e\in X. 
\hist{e}\subseteq X$.
} 
Given a 
 subset $X\subseteq E$ of events, we say that $X$ is \emph{conflict free} iff 
for all $e, e'\in X$ it holds that 
$e\neq e'\ \Rightarrow\ \neg(e\ \#\ e')$, and we denote it with $\CF{X}$. 
Given $X\subseteq E$ such that $\CF{X}$ and $Y\subseteq X$, then also $\CF{Y}$.

When adding reversibility to \pes es, conflict heredity may not hold.
Therefore, we rely on a weaker form of \pes\ by following  the approach in~\cite{PU:jlamp15}.
\begin{definition}\label{de:pre-pes}
   A \emph{pre-\pes} (p\pes) is a triple $P = (E, <, \#)$, where
   \begin{itemize}
    \item $E$ is a set of \emph{events},
    \item $\#\ \subseteq E\times E$ is an irreflexive and symmetric relation,
    \item $<\ \subseteq E\times E$ is an irreflexive partial order such that
for every $e\in E$. $\setcomp{e'\in E}{e' < e}$ is finite and
                  conflict free, and
    \item $\forall e, e'\in E$. if $e < e'$ then not $e\ \#\ e'$.
   \end{itemize}
  \end{definition}
{A p\pes\ is  a prime event structure in which conflict heredity does not hold, 
and since every \pes\ is also a p\pes\ the notions and results stated below 
for p\pes es also apply to \pes es.
}

\begin{definition}\label{de:ppes-conf-enab}
   Let $P = (E, <, \#)$ be a p\pes\ and
   $X\subseteq E$ such that $\CF{X}$. 
   For $A\subseteq E$, we say that $A$ is \emph{enabled}
   at $X$ if
   \begin{itemize}
    \item $A\cap X = \emptyset$ and $\CF{X\cup A}$, and
    \item $\forall e\in A$. if $e' < e$ then $e'\in X$.
   \end{itemize}
   If $A$ is enabled at $X$, then $X \stackrel{A}{\longrightarrow} Y$ where $Y = X\cup A$.
  \end{definition}
\begin{definition}\label{de:rpes-forwconf}
   Let $P = (E, <, \#)$ be a p\pes\ and
   {$X\subseteq E$ s.t. $\CF{X}$}.
   $X$ is a \emph{forwards reachable configuration} if there exists a sequence
   {$A_1,\ldots,A_n$}, such that
   \[ X_i\stackrel{A_i}{\longrightarrow} X_{i+1} \ \mathit{ for all}\ i, \ \mathit{and}\
       X_1 = \emptyset \ \mathit{and}\ X_{n+1}=X.
       \]
We write $\Conf{P}{p\pes}$ for the set of all  (forwards reachable) configurations of $P$.
  \end{definition}
When a p\pes\ is a \pes\ we shall write $\Conf{P}{\pes}$ instead of $\Conf{P}{p\pes}$, with 
$\Conf{P}{\pes}=\Conf{P}{p\pes}$ holding.
From a p\pes\ a \pes\ can be obtained.

\begin{definition}\label{de:hereditary-closure}
 Let $P = (E, <, \#)$ be a p\pes. Then $\mathsf{hc}(P) = (E, <, \sharp)$ is the 
 \emph{hereditary closure} of $P$, where $\sharp$ is
 derived by using the following rules 
 \[ 
\begin{array}{ccccc}
\infer{e\ \sharp\ e'}{e\ \#\ e'} & \qquad & \infer{e\ \sharp\ e''}{e\ \sharp\ e'\ \ \ \ \ e' < e''}
 & \qquad  & \infer{e\ \sharp\ e'}{e'\ \sharp\ e}\\
\end{array}
\]
\end{definition}
The following proposition relates p\pes\ to \pes\ \cite{PU:jlamp15}.
\begin{proposition}\label{pr:ppes-prop}
 Let $P = (E, <, \#)$ be a p\pes, then
 \begin{itemize}
   \item $\mathsf{hc}(P) = (E, \leq, \sharp)$ is a \pes,
   \item if $P$ is a \pes, then $\mathsf{hc}(P) = P$, and
   \item $\Conf{P}{p\pes} = \Conf{\mathsf{hc}(P)}{\pes}$.
 \end{itemize}
\end{proposition}

\subsection{Reversible prime event structures} 
We now focus on the notion of 
\emph{reversible prime event structure}. The definitions and the results in this subsection 
are drawn from \cite{PU:jlamp15}. In reversible event structures some 
events are categorised as \emph{reversible}. The added relations
are among events and those representing the \emph{actual} undoing of the reversible events.
The undoing of events is represented by \emph{removing} them (from a configuration), and is achieved
by \emph{executing} the appropriate \emph{reversing} events.

\begin{definition}\label{de:rpes}
 A \emph{reversible prime event structure} (r\pes) is the tuple 
 $\mathsf{P} = (E, \anR, <, \#, \prec, \triangleright)$ where 
 $(E, <, \#)$ is a p\pes, $\anR\subseteq E$ are the 
 \emph{reversible/undoable}
  events
 (with reverse events being denoted by 
 {$\underline{\anR} = \setcomp{\underline{\anr}}{\anr\in\anR}$ and disjoint from $E$, i.e., $\un \anR \cap E = \emptyset$})
and
 { \begin{itemize}
   \item $\triangleright\ \subseteq E\times \underline{\anR}$ is the \emph{prevention} relation,
   \item $\prec\ \subseteq E\times \underline{\anR}$ is the \emph{reverse causality} relation and
         it is such that $\anr\prec \underline{\anr}$ for each $\anr\in\anR$ and 
         $\setcomp{e\in E}{e\prec\underline{\anr}}$ is finite and conflict-free for every $\anr\in\anR$, 
   \item if $e \prec \underline{\anr}$ then not $e \triangleright \underline{\anr}$,
   \item the 
         \emph{sustained causation}  $\ll$  is a transitive relation defined such that $e \ll e'$ if $e < e'$ and if $e\in\anR$, then 
         $e' \triangleright \underline{e}$, and 
   \item $\#$ is hereditary with respect to $\ll$: if $e\ \#\ e' \ll e''$, then $e\ \#\ e''$.
 \end{itemize}
}
\end{definition}
The ingredients of a r\pes\ partly overlap with those of a \pes: there is a causality
relation ($<$) and a conflict one ($\#$) and the two are related by the \emph{sustained causation} relation $\ll$. The new ingredients are the 
\emph{prevention} relation and the \emph{reverse causality} relation.
The prevention relation states that certain events should be absent when trying to reverse an event, e.g., 
{$e\triangleright \underline{\anr}$}
states that  {$e$} should be absent when reversing {$\anr$}.  The reverse causality relation
{$e \prec \underline{\anr}$} says that 
{$\underline{\anr}$} can be executed 
only when $e$ is present.

\begin{example}\label{ex:cr not causal}
Let $\mathsf{P} = (E,\anR,<,\#,\prec,\pr)$ where
$E = \anR = \{a,b,c\}$, $a<b$ and $a \prec \un a$, $b \prec \un b$, $c \prec \un c$, $c \prec \un a$ 
with $b \pr \un a$ and no conflict.
Then $a \ll b$ because $a < b$ and $b \pr \un a$.
$\mathsf{P}$ states that $b$ causally depends on $a$ and that $c$ is concurrent w.r.t. both $a$ and $b$. 
Note that  every event is reversible in $\mathsf{P}$ because $\anR = E$.  As expected, the reverse causality relation $\prec$
is defined such that every reverse event requires the presence of the corresponding reversible event, i.e., $e\prec \un e$ for all $e\in E$. 
Additionally, it also requires $c \prec \un a$, i.e.,  $a$ can be reversed only when $c$ is present. The prevention 
relation states that $a$ cannot be reversed when $b$ is present, i.e., $b \pr \un a$.
\end{example}
%
%
\begin{definition}\label{de:rpes-conf-enab}
 Let $\mathsf{P} = (E, \anR, <, \#, \prec, \triangleright)$ be an r\pes\ and
  $X\subseteq E$ be a set of events such that $\CF{X}$. For
 $A\subseteq E$ and $B\subseteq \anR$, we say that $A\cup \underline{B}$ is \emph{enabled}
 at $X$ if
 \begin{itemize}
  \item $A\cap X = \emptyset$, $B \subseteq X$ and $\CF{X\cup A}$, 
  \item $\forall e\in A, e'\in E$. if $e' < e$ then $e'\in X\setminus B$,
  \item $\forall e\in B, e'\in E$. if $e' \prec \underline{e}$ then 
        $e' \in X\setminus (B\setminus\setenum{e})$,
  \item $\forall e\in B, e'\in E$. if $e' \triangleright \underline{e}$ then $e'\not\in X\cup A$.      
 \end{itemize}
 If $A\cup \underline{B}$ is enabled at $X$ then 
 $X \stackrel{A\cup\underline{B}}{\longrightarrow} Y$ where $Y = (X\setminus B)\cup A$.
\end{definition}
\begin{example}\label{ex:rpes-red}
Consider  the r\pes\ in Example~\ref{ex:cr not causal}.
We have, e.g.,  $\emptyset \stackrel{\{a,c\}}{\longrightarrow} \{a,c\} 
 \stackrel{\{\un a\}}{\longrightarrow} \{c\}$ and $\emptyset \stackrel{\{a\}}{\longrightarrow} \{a\}
 \stackrel{\{b\}}{\longrightarrow} \{a, b\} 
 \stackrel{\{c, \un b\}}{\longrightarrow} \{a,c\} \stackrel{\{b\}}{\longrightarrow} \{a,b,c\}$. 
 \end{example}
Reachable configurations are subsets of events which can be reached from the empty set by performing
events or undoing previously performed events.
%
%
\begin{definition}\label{de:rpes-conf}
Let $\mathsf{P} = (E, \anR, <, \#, \prec, \triangleright)$ be a r\pes\ and
 let $X\subseteq E$ be a subset of events such that $\CF{X}$.
 We say that $X$ is a \emph{(reachable) configuration} if there exist two sequences
 of 
 sets $A_i$ and $B_i$, for $i=1,\ldots,n$,  such that
 \begin{itemize}
  \item $A_i\subseteq E$ and $B_i\subseteq U$ for all $i$, and
  \item $X_i\stackrel{A_i\cup\underline{B_i}}{\longrightarrow} X_{i+1}$ for all $i$ with 
   $X_1 = \emptyset$ and $X_{n+1}=X$.
 \end{itemize}
The set of configurations of $\mathsf{P}$ is denoted by $\Conf{\mathsf{P}}{r\pes}$.
\end{definition}
\begin{example} 
The  set of configurations of $\mathsf{P}$ defined in Example~\ref{ex:cr not causal} is 
$\Conf{\mathsf{P}}{r\pes} = \{\emptyset, \{a\}, \{c\}, \{a,b\}, \{a,c\}, \{a,b,c\}\}$
as illustrated by the sequences shown in Example~\ref{ex:rpes-red}.
 \end{example}


As discussed in Section~\ref{sec:intro}, r\pes es accommodate different flavours of 
reversibility. Hereafter, we  focus on \emph{causal} reversibility \cite{Lanese14}, 
which is one of the most common models of 
reversibility in distributed systems, in which an event can reversed 
only when all the events it has caused have already been  reversed.
  
  \begin{definition}\label{de:cr-and-causal-rpes}
   Let $\mathsf{P} = (E, \anR, <, \#, \prec, \triangleright)$ be an r\pes. Then $\mathsf{P}$ is 
   \emph{cause-respecting} if for any $e, e'\in E$, if $e < e'$ then $e \ll e'$. 
   $\mathsf{P}$ is \emph{causal} if for any $e\in E$ and $\anr\in\anR$ the following holds:
	 $e \prec \underline{\anr}$ iff  $e = \anr$, and  
    $e \triangleright \underline{\anr}$ iff $\anr < e$.
  \end{definition}
%
%
  %
\begin{example}\label{ex:cr not causal2}
The r\pes\ $\mathsf{P}$ in Example~\ref{ex:cr not causal} is a cause-respecting r\pes.
However $\mathsf{P}$ is not causal because of $c \prec \un a$, which means that $c$ has to be present 
for $a$ to be reversed even if $c$ does not causally depend on $a$.  If we remove $c \prec \un a$ then we obtain a causal r\pes.
\end{example}

\begin{example}\label{out-of-causal-order}
An example of out-of-causal order reversibility can be obtained from the definition of 
the r\pes\ $\mathsf{P}$ in 
Example~\ref{ex:cr not causal} by replacing $b \pr \un a$ by 
$a \pr \un b$. Then, we have $\emptyset \stackrel{\{a\}}{\longrightarrow} \{a\}
\stackrel{\{b,c\}}{\longrightarrow} \{a,b,c\}\stackrel{\{\un a\}}{\longrightarrow} \{b,c\}$.
 Note that $a$ can be reversed even in the presence of the event $b$, which causally 
depends on $a$.
\end{example}

Cause-respecting and causal r\pes es enjoy the following useful properties \cite{PU:jlamp15}.
 
\begin{proposition}\label{pr:rpes-causal-cause-respecting}
Let $\mathsf{P} = (E, \anR, <, \#, \prec, \triangleright)$ be a r\pes.
Let $X$ be a left-closed and conflict-free subset of events in $E$ and
let $A, B\subseteq \anR$. Then
     \begin{itemize}
\item if $\mathsf{P}$ is cause-respecting and $X \stackrel{A\cup \underline{B}}{\longrightarrow} X'$, 
      then $X'$ is also left-closed,
\item if $\mathsf{P}$ is cause-respecting and
            $X \stackrel{\underline{B}}{\longrightarrow} X'$, then
            $X' \stackrel{B}{\longrightarrow} X$,
\item if $\mathsf{P}$ is causal and
            $X \stackrel{A\cup \underline{B}}{\longrightarrow} X'$, then
            $X' \stackrel{B\cup\underline{A}}{\longrightarrow} X$.
     \end{itemize}
\end{proposition}

\begin{example}\label{not-left-closed}
The above properties do not hold when a r\pes\ is not cause-respecting / causal. Consider the 
r\pes\ in Example~\ref{out-of-causal-order}. We have that $\{a,b,c\}\stackrel{\{\un a\}}{\longrightarrow} \{b,c\}$ but 
$\{b,c\}$ is not left-closed.
\end{example}

A particular r\^ole will be played by the configurations that can be reached without
executing any reversible events.
\begin{definition}\label{de:rpes-forwconf}
 Let $\mathsf{P} = (E, \anR, <, \#, \prec, \triangleright)$ be a r\pes\ and
  $X\in\Conf{\mathsf{P}}{r\pes}$ be a configuration.
 $X$ is \emph{forwards reachable} if there exists a sequence
 of sets $A_i\subseteq E$, for $i=1,\ldots,n$, such that 
 {$
 X_i\stackrel{A_i}{\longrightarrow} X_{i+1}$ for all  $i$, with  
     $X_1 = \emptyset$ and $X_{n+1}=X$.
}
\end{definition}
Although $\{b,c\}$ in Example~\ref{not-left-closed} is a reachable configuration it is 
not forwards reachable. However, the configurations of a cause-respecting r\pes\ are forwards reachable
\cite{PU:jlamp15}.
\begin{proposition}
     Let $\mathsf{P} = (E, \anR, <, \#, \prec, \triangleright)$ be a cause-respecting r\pes, and
     let $X$ be configuration of $\mathsf{P}$. Then $X$ is forwards reachable.
\end{proposition}


\section{Occurrence nets and prime event structures}\label{sec:cn-and-pes}
We now review the notion of occurrence nets~\cite{NPW:PNES,Win:ES}. 
Given a net $N = \langle S, T, F, \mathsf{m}\rangle$,
we 
write $<_N$ for transitive closure of $F$, and $\leq_N$ for the reflexive closure of $<_N$.
 We say $N$ is \emph{acyclic} if  $\leq_N$ is a partial order.
%
For occurrence nets, we adopt the usual convention and refer to places and transitions respectively 
 as {\em conditions} and {\em events}, and correspondingly use $B$  and 
$E$ for the sets of conditions and events.
 We will often confuse conditions with places and 
events with transitions.
\begin{definition}
 An \emph{occurrence net} (\cn) $C = \langle B, E, F, \mathsf{c}\rangle$ is an acyclic, safe net 
 satisfying the
  following restrictions:
  \begin{itemize}
    \item $\forall b\in \mathsf{c}$. $\pre{b} = \emptyset$,
    \item $\forall b\in B$. $\exists b'\in \mathsf{c}$ such that $b' \leq_C b$,
    \item $\forall b\in B$. $\pre{b}$ 
          is either empty or a singleton,
    \item for all $e\in E$ the set $\hist{e} = \setcomp{e' \in E}{e'\leq_C e}$ is finite,
          and
    \item  $\#$ is an irreflexive and symmetric relation defined as follows:
           \begin{itemize}
             \item $e\ \#_0\ e'$ iff $e, e' \in E$, $e\neq e'$ and 
                   $\pre{e}\cap\pre{e'}\neq \emptyset$,
             \item $x\ \#\ x'$ iff $\exists y, y'\in E$ such that $y\ \#_0\ y'$ and $y \leq_C x$ and 
                   $y' \leq_C x'$.
             \end{itemize}
     \end{itemize}
\end{definition}
The intuition behind occurrence nets is the following: each condition $b$ represents 
the occurrence of a token,
which is produced by the \emph{unique} event in $\pre{b}$, 
unless $b$ belongs to the initial marking,
and it is used by only one transition (hence if $e, e'\in\post{b}$, then $e\ \#\ e'$).
On an occurrence net $C$ it is natural to define a notion of \emph{causality} among elements of the 
net: we say that $x$ is \emph{causally dependent} on $y$ iff $y \leq_C x$.
%
%
%
 
Occurrence nets are often the result of the \emph{unfolding} of a (safe) net. 
In this perspective an occurrence net is meant to describe precisely the non-sequential semantics
of a net, and each reachable marking of the occurrence net corresponds to a reachable marking
in the net to be unfolded. Here we focus purely on occurrence nets and not on the nets
they are unfoldings of.

\begin{definition}\label{de:cn-conf}
 Let $C = \langle B, E, F, \mathsf{c}\rangle$ be a \cn\ and 
 $X\subseteq E$ be a subset of events. Then $X$ is a \emph{configuration} of
 $C$ whenever $\CF{X}$ and $\forall e\in X$. $\hist{e}\subseteq X$.
 The set of configurations of the occurrence net $C$ is denoted by $\Conf{C}{\cn}$.
\end{definition}

Given an occurrence net $C = \langle B, E, F, \mathsf{c}\rangle$ and a state
$X \in \states{C}$, it is easy to see that it is \emph{conflict free}, \emph{i.e.}
$\forall e, e'\in X$. $e\neq e'\ \Rightarrow\ \neg (e\ \#\ e')$, and \emph{left closed},
\emph{i.e.} $\forall e \in X$. $\setcomp{e'\in E}{e'\leq_C e}\subseteq X$.

\begin{proposition}\label{pr:states-are-conf}
 Let $C = \langle B, E, F, \mathsf{c}\rangle$ be an occurrence net and $X\in \states{C}$. Then
 $X\in \Conf{C}{\cn}$.
\end{proposition}
We now recall the connection among occurrence nets and prime event structures~\cite{Win:ES}.
\begin{proposition}\label{pr:on-to-pes}
 Let $C = \langle B, E, F, \mathsf{c}\rangle$ be an occurrence net. Then
 $\mathcal{P}(C) = (E, \leq_C, \#)$
  is a \pes, and 
 $\Conf{C}{\cn} = \Conf{\mathcal{P}(C)}{\pes}$.
\end{proposition}

\begin{example}\label{ex:cn}

Figure~\ref{fig:occ-nets} illustrates some  (finite) occurrence nets. We can associate \pes es to them as follows.
The net $C_1$ has two concurrent events, which hence are neither causally ordered nor in conflict,
consequently both $<$ and $\#$ are  empty. The two events $e_1$ and $e_2$ in  $C_2$ are in conflict,
namely $e_1\ \#\ e_2$, while they are causally ordered in $C_3$, namely $e_1 < e_2$, and not in conflict.
Finally, in $C_4$ we have $e_1< e_3$ and $e_2<e_4$ and $e_1\ \#\ e_2$. Additionally, conflict inheritance 
give us $e_1\ \#\ e_4$, $e_2\ \#\ e_3$ and $e_3\ \#\ e_4$.

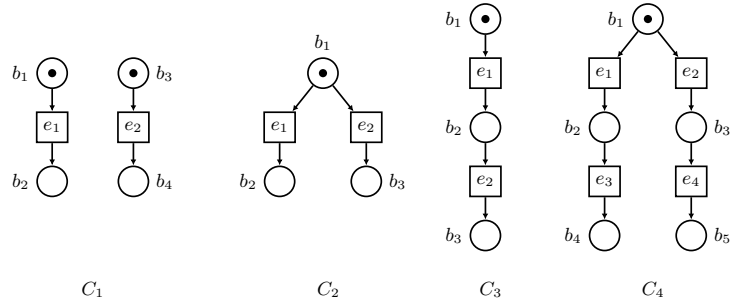
\begin{figure}[t]
\begin{center}
\scalebox{0.8}{
\begin{tikzpicture}
[bend angle=45, scale=.9, 
  pre/.style={<-,shorten
    <=0pt,>=stealth,>={Latex[width=1mm,length=1mm]},thick}, post/.style={->,shorten
    >=0,>=stealth,>={Latex[width=1mm,length=1mm]},thick},place/.style={circle, draw=black,
    thick,minimum size=5mm}, transition/.style={rectangle, draw=black,
    thick,minimum size=5mm}, invplace/.style={circle,
    draw=black!0,thick}]
\node[place] (b10) at (8,0) [label=left:$b_3$] {};
\node[place] (b14) at (10.2,0) [label=left:$b_4$] {};
\node[place] (b15) at (11.8,0) [label=right:$b_5$] {};
\node[place] (b2) at (0,1) [label=left:$b_2$] {};
\node[place] (b4) at (1.5,1) [label=right:$b_4$] {};
\node[place] (b6) at (4.2,1)  [label=left:$b_2$] {};
\node[place] (b7) at (5.8,1)  [label=right:$b_3$]{};
\node[place] (b9) at (8,2) [label=left:$b_2$] {};
\node[place] (b12) at (10.2,2) [label=left:$b_2$] {};
\node[place] (b13) at (11.8,2) [label=right:$b_3$] {};
\node[place,tokens=1] (b1) at (0,3) [label=left:$b_1$] {};
\node[place,tokens=1] (b3) at (1.5,3) [label=right:$b_3$] {};
\node[place,tokens=1] (b5) at (5,3)  [label=above:$b_1$] {};
\node[place,tokens=1] (b8) at (8,4) [label=left:$b_1$] {};
\node[place,tokens=1] (b11) at (11,4) [label=left:$b_1$] {};
\node[transition] (e1) at (0,2) {$e_1$}
edge[pre] (b1)
edge[post] (b2);
\node[transition] (e2) at (1.5,2) {$e_2$}
edge[pre] (b3)
edge[post] (b4);
\node[transition] (e3) at (4.2,2) {$e_1$}
edge[pre] (b5)
edge[post] (b6);
\node[transition] (e4) at (5.8,2) {$e_2$}
edge[pre] (b5)
edge[post] (b7);
\node[transition] (e5) at (8,3) {$e_1$}
edge[pre] (b8)
edge[post] (b9);
\node[transition] (e6) at (8,1) {$e_2$}
edge[pre] (b9)
edge[post] (b10);
\node[transition] (e7) at (10.2,3) {$e_1$}
edge[pre] (b11)
edge[post] (b12);
\node[transition] (e8) at (11.8,3) {$e_2$}
edge[pre] (b11)
edge[post] (b13);
\node[transition] (e9) at (10.2,1) {$e_3$}
edge[pre] (b12)
edge[post] (b14);
\node[transition] (e9) at (11.8,1) {$e_4$}
edge[pre] (b13)
edge[post] (b15);
\node (n1) at (.75,-1) {$C_1$};
\node (n1) at (5.1,-1) {$C_2$};
\node (n1) at (8.1,-1) {$C_3$};
\node (n1) at (11.1,-1) {$C_4$};
\end{tikzpicture}}
\caption{Some occurrence nets}
\label{fig:occ-nets}
\end{center}
\end{figure}

\end{example}

Conversely, every \pes\ can be associated with an occurrence net. 

\begin{proposition}\label{pr:pes-to-on}
 Let $P = (E, <, \#)$ be a \pes\ and let $\bot\not\in E$ be a new symbol. 
 Then $\mathcal{E}(P) = \langle B, E, F, \mathsf{c}\rangle$ 
 defined as follows 
  \begin{itemize}
    \item $B = \setcomp{(a,A)}{a\in E\cup\setenum{\bot} \land A \subseteq E \land \CF{A} \land 
    a \neq \bot\ \Rightarrow\ \forall e\in A.\ a < e}$,
   \item $F = \setcomp{(b,e)}{b = (a,A)\ \land\ e\in A}\cup \setcomp{(e,b)}{b = (e,A)}$, and
   \item $\mathsf{c} = \setcomp{(\bot,A)}{A \subseteq E\ \land\ \CF{A}}$,                                               
 \end{itemize}
 is an occurrence net, and $\Conf{P}{\pes} = \Conf{\mathcal{E}(P)}{\cn}$.
\end{proposition}


\section{Reversible causal nets and reversible prime event structures}\label{sec:rcn-and-rpes}
We now introduce 
the notion of \emph{reversible causal nets}.
A similar notion has been proposed in~\cite{MMU:coordination19} for adding
 causally-consistent reversibility to Petri nets by making  reversible every event in
the unfolding of the net. 
In this work we deal with a generalised version of 
reversible causal nets in which
 transitions may be irreversible, i.e.,  
we do not require every transition of a net to be undoable.

The intuition behind reversible causal nets is the following: we add special
transitions (events in the classical occurrence net terminology) to an occurrence net which, 
when executed, \emph{undo} 
the execution of other (standard) transitions. When we remove these special transitions from a
reversible causal net we obtain a standard occurrence net.

\begin{definition}\label{def:rcn}
 A \emph{reversible causal net} (\rcn)  is a tuple $R = \langle B, E, \Er, F, \mathsf{c}\rangle$ 
where $\langle B, E, F, \mathsf{c}\rangle$ is a safe net such that 
 \begin{itemize}
   \item $\Er\subseteq E$ and $\forall \er\in \Er$. $\exists!\ e\in E\setminus\Er$ such that
         $\pre{\er} = \post{e}$ and $\post{\er} = \pre{e}$, 
   \item $\forall e, e'\in E$. $\pre{e} = \pre{e'}\ \land\ \post{e} = \post{e'}\ 
         \Rightarrow\ e = e'$, 
   %
   \item $\bigcup_{e\in E} (\pre{e}\cup\post{e}) = B$, and        
   \item $C_{E\setminus \Er} = \langle B, E\setminus \Er, F', \mathsf{c}\rangle$ is an occurrence net, 
         where $F'$ is the restriction of $F$ to the transitions in $E\setminus \Er$.
 \end{itemize}        
\end{definition}
The events in $\Er$ are the reversing ones and
we often say that a reversible causal net $R$ is reversible \emph{with respect to} $\Er$. 
We  write $\overline{E}$ 
 for the set of events $E\setminus \Er$ and
$C_{\overline{E}}$ instead of $C_{E\setminus \Er}$.
The first condition in Definition~\ref{def:rcn} implies that 
 each reversing event $\er\in\Er$ is associated with  a unique event $e$ that 
causes the effects that  $\er$ is intended to \emph{undo}; hence $e$ here is a
\emph{reversible} event. Moreover, 
the second condition ensures that there is an injective
mapping $h : \Er \to E$ that associates  each event  ${\er\in\Er}$ with a different 
event $e\in E$ such that $\pre{e} = \post{\er}$ and $\post{e} = \pre{\er}$, in other words, each reversible event has exactly one reversing event.
The third requirement guarantees that all conditions (places) of the net appear at least in the pre or the postset of some event  (transitions), i.e., 
there are no isolated conditions. 
The last condition ensures that the net obtained by deleting all reversing events is an
occurrence net.
%

\begin{example}\label{ex:rcn}
 We present some reversible causal nets in Figure~\ref{fig:rcn}. The reversing events are drawn in red, and
 their names are underlined.
  The events $e_1$ and $e_2$ in $R_1$ are both reversible, while 
  $e_1$ is the only  reversible event in $R_2$. In $R_3$ the events $e_1$, $e_3$ and $e_4$ are the reversible ones.

 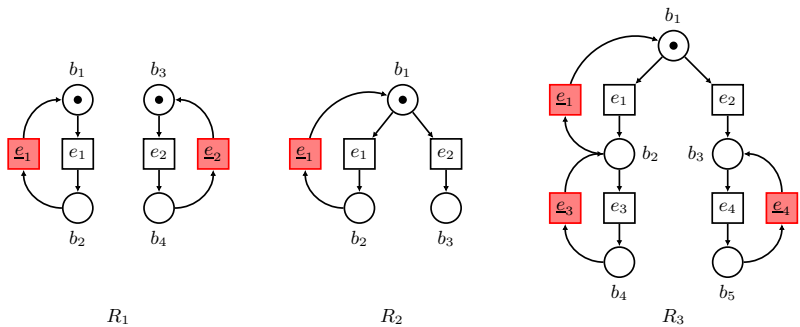
\begin{figure}[t]
 \begin{center}
  \scalebox{0.8}{
\begin{tikzpicture}
[bend angle=45, scale=.9, 
  pre/.style={<-,shorten
    <=0pt,>=stealth,>={Latex[width=1mm,length=1mm]},thick}, post/.style={->,shorten
    >=0,>=stealth,>={Latex[width=1mm,length=1mm]},thick}, rpost/.style={=>,shorten
    >=0,>=stealth,>={Latex[width=1mm,length=1mm]},thick}, place/.style={circle, draw=black,
    thick,minimum size=5mm}, transition/.style={rectangle, draw=black,
    thick,minimum size=5mm}, revtransition/.style={rectangle, draw=red,
    thick,fill=red!50,minimum size=5mm}, invplace/.style={circle,
    draw=black!0,thick}]
\node[place] (b14) at (9,0) [label=below:$b_4$] {};
\node[place] (b15) at (11,0) [label=below:$b_5$] {};
\node[place] (b2) at (-1,1) [label=below:$b_2$] {};
\node[place] (b4) at (0.5,1) [label=below:$b_4$] {};
\node[place] (b6) at (4.2,1)  [label=below:$b_2$] {};
\node[place] (b7) at (5.8,1)  [label=below:$b_3$]{};
\node[place] (b12) at (9,2) [label=right:$b_2$] {};
\node[place] (b13) at (11,2) [label=left:$b_3$] {};
\node[place,tokens=1] (b1) at (-1,3) [label=above:$b_1$] {};
\node[place,tokens=1] (b3) at (0.5,3) [label=above:$b_3$] {};
\node[place,tokens=1] (b5) at (5,3)  [label=above:$b_1$] {};
\node[place,tokens=1] (b11) at (10,4) [label=above:$b_1$] {};
\node[transition] (e1) at (-1,2) {$e_1$}
edge[pre] (b1)
edge[post] (b2);
\node[transition] (e2) at (0.5,2) {$e_2$}
edge[pre] (b3)
edge[post] (b4);
\node[revtransition] (re1) at (-2,2) {$\underline{e}_1$}
edge[pre, bend right] (b2)
edge[post, bend left] (b1);
\node[revtransition] (re2) at (1.5,2) {$\underline{e}_2$}
edge[pre, bend left] (b4)
edge[post, bend right] (b3);
\node[transition] (e3) at (4.2,2) {$e_1$}
edge[pre] (b5)
edge[post] (b6);
\node[revtransition] (re3) at (3.2,2) {$\underline{e}_1$}
edge[pre, bend right] (b6)
edge[post, bend left] (b5);
\node[transition] (e4) at (5.8,2) {$e_2$}
edge[pre] (b5)
edge[post] (b7);
\node[transition] (e7) at (9,3) {$e_1$}
edge[pre] (b11)
edge[post] (b12);
\node[transition] (e8) at (11,3) {$e_2$}
edge[pre] (b11)
edge[post] (b13);
\node[revtransition] (re7) at (8,3) {$\underline{e}_1$}
edge[pre, bend right] (b12)
edge[post, bend left] (b11);
\node[transition] (e9) at (9,1) {$e_3$}
edge[pre] (b12)
edge[post] (b14);
\node[revtransition] (re9) at (8,1) {$\underline{e}_3$}
edge[pre, bend right] (b14)
edge[post, bend left] (b12);
\node[transition] (e9) at (11,1) {$e_4$}
edge[pre] (b13)
edge[post] (b15);
\node[revtransition] (re9) at (12,1) {$\underline{e}_4$}
edge[pre, bend left] (b15)
edge[post, bend right] (b13);
\node (n1) at (-.25,-1) {$R_1$};
\node (n2) at (4.8,-1) {$R_2$};
\node (n3) at (10,-1) {$R_3$};
\end{tikzpicture}}
 \end{center}
 \caption{Some reversible causal nets}
 \label{fig:rcn}
 \end{figure}
\end{example}

We prove that the set of reachable markings of a reversible causal net is not influenced 
by performing reversing events.
\begin{proposition}\label{pr:reachable-markings-of-rcn}
 Let $R = \langle B, E, \Er, F, \mathsf{c}\rangle$ be a \rcn.
 Then $\reachMark{R} = \reachMark{C_{\overline{E}}}$.
\end{proposition}
A consequence of the above proposition is the following corollary, which 
establishes
that each marking can be reached by using just \emph{forward events}.
\begin{corollary}\label{cor:a}
 Let $C = \langle B, E, \Er, F, \mathsf{c}\rangle$ be a \rcn\
 and $\sigma$ be a \fs. Then, there exists a \fs\ $\sigma'$ such that
 $X_{\sigma'}\subseteq \overline{E}$ and $\lead{\sigma} = \lead{\sigma'}$.
\end{corollary}

\begin{definition}\label{de:conf-rcn}
 Let $R = \langle B, E, \Er, F, \mathsf{c}\rangle$ be a \rcn, and
  $X \subseteq \overline{E}$ be a subset of forward events. 
 Then,  $X$ is a configuration of $R$ iff $X$ is a configuration of $C_{\overline{E}}$.
 The set of configurations of $R$ is as usually denoted with $\Conf{R}{\rcn}$.
\end{definition}
A configuration of a reversible causal net $R$ with respect to $\Er$ is a subset of  
$E\setminus \Er$;
consequently, the reversing events (i.e., the ones in $\Er$) that may have been executed to reach a particular 
marking are not considered as part of the configuration.
Observe that, differently from 
occurrence net, $\states{R} \neq \Conf{R}{\rcn}$ because the former may contain also reversing events.
However, as a consequence of Corollary \ref{cor:a}, there is no loss of information.

We show how to construct a reversible causal net from an occurrence net, once we have
identified the events to be \emph{reversed}.

\begin{definition}\label{de:constructing-a-reversible-causal-net}
 Let $C = \langle B, E, F, \mathsf{c}\rangle$ be an occurrence net and $\Er \subseteq E$ be a set 
 of reversible events. Define $\mathsf{R}(C) = \langle B, \hat{E}, \Er\times\setenum{r}, \hat{F}, \mathsf{c}\rangle$ 
 be the net where $\hat{E}$ and $\hat{F}$ are defined as follows:
  \begin{itemize}
   \item $\hat{E} = E\times\setenum{\fe} \ \cup\  \Er\times\setenum{\re}$, and
   \item 
$\hat{F} = 
		\setcomp{(b, (e,\fe))}{(b,e)\in F} \quad\cup\quad \setcomp{ ((e,\fe),b)}{(e,b)\in F}\quad\cup$ \phantom{asdfasdfas} \phantom{asdfa}
		$\setcomp{(b, (e,\re))}{(e,b)\in F}\quad\cup\quad \setcomp{ ((e,\re),b)}{(b,e)\in F}$
%
%
  \end{itemize}
  The mapping $h : \Er\times\setenum{\re} \to E\times\setenum{\fe}$ is defined as
  $h(e,\re) = (e,\fe)$. 
\end{definition}
The construction above simply adds as many events (transitions) as those to
be reversed. The preset of each added event is the postset of the
corresponding event to be reversed, and its postset is defined as the preset of the event to 
be reversed.
The events in $\Er\times\setenum{\re}$ are the reversing events.

\begin{proposition}\label{prop:rev_net}
 Let $C = \langle B, E, F, \mathsf{c}\rangle$ be an occurrence net, $\Er\subseteq E$ be the subset of
 reversible events, and 
 $\mathsf{R}(C) = \langle B, \hat{E}, \Er\times\setenum{\re}, \hat{F}, \mathsf{c}\rangle$ be the net 
  in Definition \ref{de:constructing-a-reversible-causal-net}. Then, $\mathsf{R}(C)$ is a reversible causal net with respect to
 $\Er\times\setenum{\re}$.
\end{proposition}
%
\begin{example}\label{ex:reversing}
 Consider the occurrence net $C_1$ in Figure~\ref{fig:occ-nets}, and assume that both events are reversible. The 
 net $R_1$ in Figure~\ref{fig:rcn} is $\mathsf{R}(C_1)$ (after renaming events with
 the convention that $(e,\fe)$ is named as $e$ and $(e,\re)$ as $\underline{e}$). The \rcn\ $R_3$ in Figure~\ref{fig:rcn} is 
 $\mathsf{R}(C_4)$, with $C_4$ in Figure~\ref{fig:occ-nets} and the set of reversible events  $\Er = \{e_1,e_2,e_4\}$.
\end{example}

\subsection{From \rcn\ to r\pes}
As it is usually done for causal nets, 
we now associate each reversible causal net with a reversible prime event structure.
Given an \rcn\ $R = \langle B, E, \Er, F, \mathsf{c}\rangle$, we denote the
set of events $\setcomp{e'}{e  <_R e'}$ by $\future{e}$. Observe that this set is not necessarily conflict-free.

\begin{proposition}\label{pr:rce-to-rpes}
 Let $R = \langle B, E, \Er, F, \mathsf{c}\rangle$ be a reversible causal net with respect to $\Er$, 
 then $\mathcal{C}_{r}(R) = (E', \anR', <, \#, \prec, \triangleright)$ is its associated r\pes, where
 \begin{itemize}
  \item $E' = \overline{E}$ and $\anR'= h(\Er)$,
  \item $<$ is 
         $<_{C_{\overline{E}}}$. 
  \item $\#$ is the conflict relation defined on the occurrence net $C_{\overline{E}}$,
  \item $e\ \triangleright\ \underline{e'}$ whenever $e\in \future{e'}$, 
  \item $e\ \prec\ \underline{e'}$ whenever $e = e'$, and
  \item $\ll = <$.         
 \end{itemize}
\end{proposition}

%
%

\begin{example}
 Consider the reversible causal net $R_3$ in Figure~\ref{fig:rcn}. The associated r\pes\ has
 the events $\setenum{e_1,e_2,e_3,e_4}$ and the reversible events $\setenum{e_1,e_3,e_4}$. 
 The causality relation of the associated p\pes\ 
 is $e_1 < e_3$, $e_2 < e_4$, the 
 conflict relation is \emph{generated} by $e_1 \# e_2$, and it is inherited along $\ll$, which 
 coincides with $<$. The reverse causality stipulates that $e_1 \prec \underline{e}_1$, 
 $e_3 \prec \underline{e}_3$ and $e_4 \prec \underline{e}_4$ and finally
 $e_3 \triangleright \underline{e}_1$, as to be allowed to undo $e_1$ it is necessary to
 undo $e_3$ first.
\end{example}

The following result states that the r\pes\ associated to a reversible causal net is 
causal, hence cause-respecting.
\begin{proposition}\label{prp:rpes_respect}
 Let $R = \langle B, E, \Er, F, \mathsf{c}\rangle$ be a reversible causal net with respect to $\Er$ and 
 $\mathcal{C}_{r}(R) = (E', \anR', <, \#, \prec, \triangleright)$ be the associated r\pes. 
 Then $\mathcal{C}_{r}(R)$ is a 
 causal r\pes.
\end{proposition}

We show that each configuration of a \rcn\ is a configuration of the corresponding
r\pes,  and vice versa.
\begin{theorem}\label{th:rnctorpes-conf-correspond}
 Let $R = \langle B, E, \Er, F, \mathsf{c}\rangle$ be a reversible causal net with respect to $\Er$ and 
 $\mathcal{C}_{r}(R) = (E', \anR', <, \#, \prec, \triangleright)$ be the associated r\pes. 
 Then $X \subseteq E'$ is a configuration of $R$ iff $X$ is a configuration of $\mathcal{C}_{r}(R)$.
\end{theorem}
%

We stress that a reversing event in a reversible causal net is enabled at a marking when
the conditions in the postset of the event to be reversed are marked. 
This may happen only
when all the events that causally depend on the event to be reversed have either been
executed and reversed or not been executed at all.
Thus every \rcn\ enjoys \emph{causally consistent} reversibility~\cite{rccs,rhotcs}, 
and consequently cannot implement the
so called \textit{out-of-causal order} reversibility~\cite{PhiUliYuen12,UK16}.
%
Contrastingly,  r\pes es are able to model \textit{out-of-causal order} reversibility (as illustrated in  Example~\ref{out-of-causal-order}).

The proposition below establishes a  correspondence between the steps in a reversible causal net
and the sequences of reachable configurations of the r\pes\  associated to that net. 
Proposition~\ref{prop:mixed_step} below formalises what is called \emph{mixed-reverse} transitions in~\cite{EKM:petri_net}. 
We now introduce some auxiliary notation. Let $R = \langle B, E, \Er, F, \mathsf{c}\rangle$ be a \rcn, and $X\subseteq E$ be a 
configuration of $R$, we write  $\mathsf{mark}(X)$ to denote the marking reached 
after executing the events
in $X$; this marking can be expressed as $(\mathsf{c}\cup\post{X})\setminus \pre{X}$.

\begin{proposition}\label{prop:mixed_step}
 Let $R = \langle B, E, \Er, F, \mathsf{c}\rangle$ be a reversible causal net and  
 $\mathcal{C}_{r}(R) = (E', \anR', <, \#, \prec, \triangleright)$ be its associated r\pes. 
 Let $X\in\Conf{R}{\rcn}$ and  $A\subseteq E$ be a set of events such that 
 $\mathsf{mark}(X)\trans{A}$. Then $\hat{A}\cup \underline{B}$ is enabled at $X$ in
 $\mathcal{C}_{r}(R)$, where $\hat{A} = \setcomp{e\in A}{e\not\in \Er}$ and
 $\underline{B} = \setcomp{e\in A}{e\in \Er}$.
\end{proposition}

\subsection{From r\pes\ to \rcn}
Correspondingly to what is usually done when relating nets to event structures, we show that 
if we focus on causal r\pes{es} then we can relate them to reversible occurrence nets.
The construction is indeed quite standard (see \cite{Win:ES,BCP:LenNets} among many others),
but we do need a further observation on causal r\pes.

\begin{proposition}\label{pr:pes-associated-to-crpes}
 Let $\mathsf{P} = (E, \anR, <, \#, \prec, \triangleright)$ be a causal r\pes\ and
 let $<^{+}$ be the transitive closure of $<$. Then, $\#$ is inherited along  $<^{+}$,
 \emph{i.e.} $e\ \#\ e' <^{+} e''\ \Rightarrow\ e\ \#\ e''$.
\end{proposition}
A consequence of this proposition is that the conflict relation is fully characterized by the
causality relation, and the same intuition for introducing reversible causal net can be used in
associating a net to a causal r\pes\ like the one used to associate an occurrence net to a \pes. 

\begin{definition}\label{de:rpestorcn}
 Let $\mathsf{P} = (E, \anR, <, \#, \prec, \triangleright)$ be a causal r\pes, and 
  $\bot\not\in E$ be a new symbol. Define 
 $\mathcal{E}_{r}(\mathsf{P})$ as the Petri net $\langle B, \hat{E}, F, \mathsf{c}\rangle$ where
 \begin{itemize}
    \item $B = \setcomp{(a,A)}{a\in E\cup\setenum{\bot} \land A \subseteq E \land \CF{A} \land 
    a \neq \bot\ \Rightarrow\ \forall e\in A.\ a \ll e}$,\smallskip
   \item $\hat{E} = E\times\setenum{\fe}\  \cup\  \anR\times\setenum{\re}$,
   \item $F = \setcomp{(b,(e,\fe))}{b = (a,A)\ \land\ e\in A} \quad \cup\quad
              \setcomp{((e,\fe),b)}{b = (e,A)}\quad\cup\quad$ \phantom{asdfa}
              $\setcomp{(b,(e,\re))}{b = (e,A)}\quad\cup\quad
              \setcomp{((e,\re),b)}{b = (a,A)\ \land\ e\in A}$,
                          and
   \item $\mathsf{c} = \setcomp{(\bot,A)}{A \subseteq E\ \land\ \CF{A}}$,                                               
 \end{itemize}
\end{definition} 
In essence the construction above takes the \pes\ associated to an r\pes\ and
constructs the associated occurrence net, which is then \emph{enriched} with 
the reversing events (transitions). The result is a reversible occurrence net.

\begin{proposition}\label{prp:rpes_to_cnet}
 Let $\mathsf{P} = (E, \anR, <, \#, \prec, \triangleright)$ be a causal r\pes. Then  
 $\mathcal{E}_{r}(\mathsf{P}) = \langle B, \hat{E}, \anR\times\setenum{\re}, F, \mathsf{c}\rangle$ as defined
 in Definition \ref{de:rpestorcn} is a reversible occurrence net with respect to $\anR\times\setenum{\re}$.
\end{proposition} 
\begin{theorem}\label{th:cccrpestorcn}
 Let $\mathsf{P}$ be a causal 
 r\pes. Then  
 $X'$ is a configuration of $\mathcal{E}_{r}(\mathsf{P})$ iff $X$ 
 is a configuration of $\mathsf{P}$, where $X' = \setcomp{(e,\fe)}{e\in X}$.
\end {theorem}
%
%
Clearly, if we start from a reversible causal net, we get a r\pes\ from which a reversible causal
net can be obtained having the same states (up to renaming of events).
\begin{corollary}
 Let $R$ be a \rcn. Then $\states{\mathcal{E}_{r}(\mathcal{C}_{r}(R))} = \states{R}$.
\end{corollary}

\begin{example}
 Consider the r\pes\ with four events $\setenum{e_1,e_2,e_3,e_4}$ such that $e_1 < e_3$ and $e_2 < e_4$,
 $e_1$ is in conflict with $e_2$ and this conflict is inherited along $<$. Furthermore, let 
 $e_1$ and $e_3$ be reversible, and $e_3 \triangleright \underline{e}_1$. The construction 
in Definition~\ref{de:rpestorcn} gives the net below.
 \begin{center}
  \newcommand{\aeuno}{$(e_1,\emptyset)$}
\newcommand{\eunoa}{$(\bot,\setenum{e_1})$}
\newcommand{\aedue}{$(e_2,\emptyset)$}
\newcommand{\eduea}{$(\bot,\setenum{e_2})$}
\newcommand{\aetre}{$(e_3,\emptyset)$}
\newcommand{\etrea}{$(\bot,\setenum{e_3})$}
\newcommand{\aequattro}{$(e_4,\emptyset)$}
\newcommand{\equattroa}{$(\bot,\setenum{e_4})$}
\newcommand{\eunoetre}{$(e_1,\setenum{e_3})$}
\newcommand{\edueequattro}{$(e_2,\setenum{e_4})$}
\newcommand{\ceunoedue}{$(\bot,\setenum{e_1,e_2})$}
\newcommand{\ceunoequattro}{$(\bot,\setenum{e_1,e_4})$}
\newcommand{\cedueetre}{$(\bot,\setenum{e_2,e_3})$}
\newcommand{\cetreequattro}{$(\bot,\setenum{e_3,e_4})$}
\scalebox{0.8}
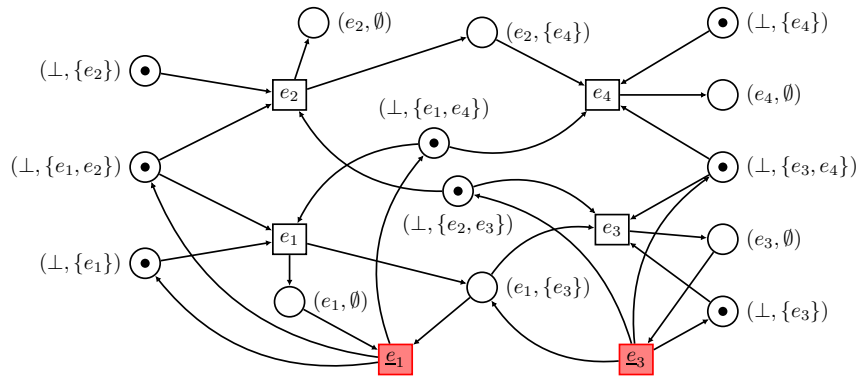
{
\begin{tikzpicture}
[bend angle=30, scale=.8, 
  pre/.style={<-,shorten
    <=0pt,>=stealth,>={Latex[width=1mm,length=1mm]},thick}, post/.style={->,shorten
    >=0,>=stealth,>={Latex[width=1mm,length=1mm]},thick}, prei/.style={o-,shorten
    <=0pt,>=stealth,>={Latex[width=1mm,length=1mm]},thick}, read/.style={-,shorten
    <=0pt,,>=stealth,>={Latex[width=1mm,length=1mm]},draw=red,very thick}, place/.style={circle, draw=black,
    thick,minimum size=5mm}, transition/.style={rectangle, draw=black,
    thick,minimum size=5mm}, revtransition/.style={rectangle, draw=red,
    thick,fill=red!50,minimum size=5mm}, invplace/.style={circle,
    draw=black!0,thick}]
\node[place,tokens=1] (b21) at (-2,5) [label=left:\eduea] {};
\node[place,tokens=1] (b12c) at (-2,3) [label=left:\ceunoedue] {};
\node[place,tokens=1] (b11) at (-2,1) [label=left:\eunoa] {};
\node[place] (b12) at (1,0.2) [label=right:\aeuno] {};
\node[place] (b22) at (1.5,6)  [label=right:\aedue] {};
\node[place] (b13cau) at (5,0.5) [label=right:\eunoetre] {};
\node[place] (b24cau) at (5,5.8)  [label=right:\edueequattro] {};
\node[place,tokens=1] (b41) at (10,6) [label=right:\equattroa] {};
\node[place,tokens=1] (b31) at (10,0) [label=right:\etrea] {};
\node[place] (b32) at (10,1.5) [label=right:\aetre] {};
\node[place] (b42) at (10,4.5)  [label=right:\aequattro] {};
\node[place,tokens=1] (b14c) at (4,3.5) [label=above:\ceunoequattro] {};
\node[place,tokens=1] (b23c) at (4.5,2.5) [label=below:\cedueetre] {};
\node[place,tokens=1] (b34c) at (10,3) [label=right:\cetreequattro] {};
\node[transition] (e1) at (1,1.5) {$e_1$}
edge[pre] (b12c)
edge[pre] (b11)
edge[post] (b12)
edge[post] (b13cau)
edge[pre, bend left] (b14c);
\node[revtransition] (re1) at (3.2,-1) {$\underline{e}_1$}
edge[post, bend left] (b12c)
edge[post, bend left] (b11)
edge[pre] (b12)
edge[pre] (b13cau)
edge[post, bend left] (b14c);
\node[transition] (e2) at (1,4.5) {$e_2$}
edge[pre, bend right] (b23c)
edge[pre] (b12c)
edge[pre] (b21)
edge[post] (b24cau)
edge[post] (b22);
\node[transition] (e3) at (7.7,1.7) {$e_3$}
edge[pre, bend right] (b23c)
edge[pre] (b34c)
edge[pre, bend right] (b13cau)
edge[pre] (b31)
edge[post] (b32);
\node[revtransition] (re3) at (8.2,-1) {$\underline{e}_3$}
edge[post, bend right] (b23c)
edge[post, bend left] (b34c)
edge[post, bend left] (b13cau)
edge[post] (b31)
edge[pre] (b32);
\node[transition] (e4) at (7.5,4.5) {$e_4$}
edge[pre, bend left] (b14c)
edge[pre] (b34c)
edge[pre] (b24cau)
edge[pre] (b41)
edge[post] (b42);
\end{tikzpicture}}

  \end{center}
\end{example}

\section{Categories of causal r\pes\ and \rcn}\label{sec:category}
Occurrence nets and \pes es are equipped with morphisms and turned  into categories that are
related by suitable functors. 
In this section, we extends such constructions to \rcn\ and causal-r\pes. 
We start by recalling the notions of morphisms for occurrence nets and prime event structures.
\begin{definition}\label{de:occ-net-morph}
 Let $C_0 = \langle B_0, E_0, F_0, \mathsf{c}_0\rangle$ and $C_1 = \langle B_1, E_1, F_1, \mathsf{c}_1\rangle$
 be two occurrence nets. Then, an occurrence net \emph{morphism} from $C_0$ to $C_1$ is a pair $(\beta,\eta)$ where
 $\beta$ is a relation between $B_0$ and $B_1$, $\eta : E_0\to E_1$ is a partial mapping such that
 \begin{itemize}
   \item $\beta(\mathsf{c}_0) = \mathsf{c}_1$ and for each $b_1\in \mathsf{c}_1$ there exists a unique
         $b_0\in \mathsf{c}_0$ such that $b_0 \beta b_1$, and
   \item for all $e_0\in E_0$, when $\eta(e_0)$ is defined and equal to $e_1$ then 
         $\beta(\pre{e_0}) = \pre{\eta(e_0)} = \pre{e_1}$
         and $\beta(\post{e_0}) = \post{\eta(e_0)} = \post{e_1}$.
 \end{itemize}
\end{definition}
As we consider just occurrence nets, we have also that if $b_0 \beta b_1$ and $(e_1,b_1)\in F_1$, then there
exists a unique event $e_0\in E_0$ such that $\eta(e_0) = e_1$ and $(e_0,b_0)\in F_0$. 
Consequently, if a condition $b_1$ in $C_1$ is related to a condition $b_0$ in $C_0$ 
by $\beta$, then the events producing these conditions are related 
by $\eta$; moreover, there is a unique event $e_0$ in $C_0$ mapped to the event $e_1$. 
Morphisms on occurrence nets compose and the identity mapping is defined as
the identity relation on conditions and the 
identity mapping on events. 
We have the category having as objects the occurrence nets and as morphisms the 
occurrence nets morphisms, which will be denoted with $\mathbf{Occ}$.

\begin{definition}\label{de:pes-morph}
 Let $P_0 = (E_0, <_0, \#_0)$ and $P_1 = (E_1, <_1, \#_1)$ be two \pes es. Then,  a
 \emph{\pes\ morphism} from $P_0$ to $P_1$ is a partial mapping $f : E_0 \to E_1$ such that
 \begin{enumerate}
   \item for all $e_0\in E_0$ if $f(e_0)$ is defined then $\hist{f(e_0)} \subseteq f(\hist{e_0})$, 
   \item for all $e_0, e_0'\in E_0$ such that  $f(e_0)$ and $f(e_0')$ are both defined and
         $f(e_0) \#_1 f(e_0')$ then $e_0 \#_0 e_0'$, and 
   \item for all $e_0, e_0'\in E_0$ such that  $f(e_0)$ and $f(e_0')$ are both defined and equal, then
         $e_0 \neq e_0'$ implies $e_0 \#_0 e_0'$.    
 \end{enumerate}
\end{definition}

Again \pes\ morphisms compose, and we have the category $\mathbf{PES}$ 
 whose objects are \pes es
and whose arrows are the \pes\ morphisms.

The  categories $\mathbf{Occ}$ and $\mathbf{PES}$ are related as follows:  each occurrence net is associated with a prime event structure 
as stated
in Proposition~\ref{pr:on-to-pes}; each occurrence net morphism
 $(\beta,\eta)$ is associated with  $\eta$, which turns out to be a \pes\ morphism. We denote such functor as
$\mathcal{P}$  (as in Proposition~\ref{pr:on-to-pes}).
Conversely, the mapping $\mathcal{E}$ introduced in Proposition~\ref{pr:pes-to-on} can be turned into a functor, 
as shown in \cite{Win:ES}.  
A \pes\ morphism $f$ is turned into an occurrence net morphism as follows:
the relation $\beta$ on conditions is defined such that (i) $(a,A) \beta (a',A')$ when $a = \bot = a'$, 
$A' = f(A) \neq \emptyset$, and $|A| = 1$, and (ii) $(e,A) \beta (f(e),A')$ when $f(e)$ is defined and
$A' = f(A) \neq \emptyset$;  the partial mapping on events is just $f$. It is then easy to
see that such definition indeed conforms an occurrence net morphism. 
The nice result is that these functors form a coreflection, where $\mathcal{C}$ is the right
adjoint and $\mathcal{E}$ the left adjoint.

We now recall the notion of r\pes\ morphisms introduced in \cite{GPY:CatRES}.

\begin{definition}\label{de:rpes-morph}
  Let $P_0 = (E_0, \anR_0, <_0, \#_0, \prec_0, \triangleright_0)$ and $P_1 = (E_1, \anR_1, <_1, \#_1, \prec_0, \triangleright_0)$ be two r\pes. Then, an
 \emph{r\pes\ morphism} from $P_0$ to $P_1$ is a partial mapping $f : E_0 \to E_1$ such that
 \begin{enumerate}
   \item for all $e_0\in E_0$ if $f(e_0)$ is defined then 
         $\setcomp{e_1\in E_1}{e_1 <_1 f(e_0)} \subseteq \setcomp{f(e)}{e <_0 e_0}$, 
   \item for all $e_0, e_0'\in E_0$ such that  $f(e_0)$ and $f(e_0')$ are both defined and
         $f(e_0) \#_1 f(e_0')$ then $e_0 \#_0 e_0'$, 
   \item for all $e_0, e_0'\in E_0$ such that  $f(e_0)$ and $f(e_0')$ are both defined and equal, then
         $e_0 \neq e_0'$ implies $e_0 \#_0 e_0'$, 
   \item for all $e_0\in E_0$ and $e\in\anR_0$ such that
         $f(e_0)$ and $f(e)$ are both defined  and $f(e_0) \triangleright_1 f(e)$ then
         $e_0 \triangleright_0 e$, and
   \item for all $e_0\in\anR_0$ if $f(e_0)$ is defined then 
         $\setcomp{e_1\in E_1}{e_1 \prec_1 \underline{f(e_0)}} \subseteq \setcomp{f(e)}{e <_0 \underline{e_0}}$.          
 \end{enumerate}
\end{definition}
The notion of morphism above generalise the one on \pes\ by requiring that prevention is preserved as
well as the reverse causality relation. 
In \cite{GPY:CatRES} it is also shown that 
r\pes\ and r\pes\ morphisms form a category, which is called $\mathbf{RPES}$. 
%
We restrict our attention to the subcategory $\mathbf{cRPES}$, which 
 has causal-r\pes as objects and r\pes\ morphisms as arrows.

\begin{proposition}\label{pr:full-subcat}
 $\mathbf{cRPES}$ is a full subcategory of $\mathbf{RPES}$.
\end{proposition}

We now turn our attention to reversible causal nets and introduce the notion of morphisms for reversible causal nets.

\begin{definition}\label{de:rcn-morph}
 Let $R_0 = \langle B_0, E_0, \Er_0, F_0, \mathsf{c}_0\rangle$ and 
 $R_1 = \langle B_1, E_1, \Er_1, F_1, \mathsf{c}_1\rangle$ be two \rcn{s}.
 Then an \rcn\ \emph{morphism} from $R_0$ to $R_1$ is the pair $(\beta,\eta)$ such that
 \begin{itemize}
   \item $(\beta,\eta)$ restricted to the occurrence nets  
         $\langle B_0, E_0\setminus \Er_0, F'_0, \mathsf{c}_0\rangle$ to
         $\langle B_1, E_1\setminus \Er_1, F'_1, \mathsf{c}_1\rangle$ is an occurrence net morphism, 
         and
   \item $\eta(\Er_0) \subseteq \Er_1$ and if $e\in \Er_0$ and $\eta(e)$ is defined then
         also $f(e')$ is defined, where $e' = h(e)$.      
 \end{itemize}
\end{definition}
It is straightforward to check that \rcn\ morphisms compose and that the identity relation and the identity mapping on events
conforms a \rcn\  morphism. Hence,  \rcn\ together with \rcn\ morphism form a category, which we call
$\mathbf{rOcc}$. 

Then, the definition of $\mathcal{C}$ in Proposition~\ref{pr:rce-to-rpes} can be \changed{easily}{} extended to a functor.
\begin{proposition}\label{pr:C-is-functor} 
 $\mathcal{C} : \mathbf{rOcc} \to \mathbf{cRPES}$ acting on objects as in Proposition~\ref{pr:rce-to-rpes}
 and on morphisms $(\beta,\eta) : R_0 \to R_1$ as $\eta$ restricted to the events that are not
 reversing ones, is a functor.
\end{proposition}

Also the construction in Definition~\ref{de:rpestorcn} can be extended to a functor. 

\begin{proposition}\label{pr:E-is-functor} 
 $\mathcal{E}_{r} : \mathbf{cRPES} \to \mathbf{rOcc}$ acting on objects as in Definition~\ref{de:rpestorcn}
 and on morphisms as stipulated in for occurrence net, requiring that reversing events are
 preserved, is a functor.
\end{proposition}

Along the lines of \cite{Win:ES} we can establish a relation between the categories
$\mathbf{cRPES}$ and  $\mathbf{rOcc}$.

\begin{theorem}\label{th:coreflection} 
 $\mathcal{E}_{r}$ and $\mathcal{C}_{r}$ form a coreflection, where $\mathcal{C}_{r}$ is the right
adjoint and $\mathcal{E}_{r}$ the left adjoint.
\end{theorem}


\section{Conclusions and future works}\label{sec:conc}
The constructions we have proposed to associate a reversible causal net to a causal 
reversible prime event structure, and vice versa, are certainly driven by the classical ones 
(see \cite{Win:ES}) for relating occurrence nets and prime event structures. 
The consequence of this approach is that the causality relation, either the one given in a r\pes\ or
the one induced by the flow relation in the occurrence net obtained ignoring the reversing
events, is the one driving the construction. One of the other two relations of an r\pes\ is substantially
ignored (and we obtain from a \rcn\ a causal r\pes\ where the reverse causality relation just 
says that an event can be reversed only after it has occurred) whereas the second (prevention)
is tightly related to the causality relation: $b$ is caused by $a$ precisely when $b$ prevents 
of undoing of $a$.
The notion of reversible causal net we have proposed suggests this construction, so the problem 
of finding which kind of net would correspond to, for example, a cause-respecting or even
an arbitrary r\pes\ remains open and
certainly deserves to be investigated.

It is however interesting to observe that the construction in Definition \ref{de:rpestorcn}
gives a reversible causal net even when the r\pes\ one started with is not a causal r\pes. 
 Consider the r\pes\ with two events $\setenum{e_1, e_2}$ such that
 $e_1 < e_2$ and where the conflict and the prevention relations are empty.
 The only reversible event is $e_1$ and $e_1 \prec \underline{e}_1$.
 The set $\{e_2\}$ is a reachable configuration: we can remove $e_1$ from a reachable configuration 
$\setenum{e_1, e_2}$ by performing the event $\underline{e}_1$. 
This is an example of out-of-causal order computation. 
Given this r\pes, our construction produces the following \rcn, which does not have $\{e_2\}$ among its configurations.
 \begin{center}
   \newcommand{\annone}{$(\bot,\setenum{e_1})$}
\newcommand{\anntwo}{$(e_1,\emptyset)$}
\newcommand{\annthree}{$(e_1,\setenum{e_2}$}
\newcommand{\annfour}{$(\bot,\setenum{e_2})$}
\newcommand{\annfive}{$(e_2,\emptyset)$}
\scalebox{0.8}
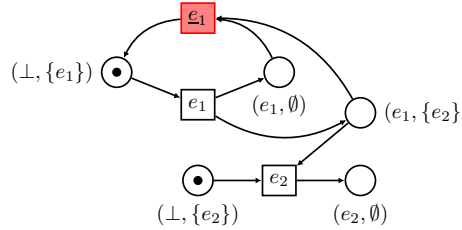
{
\begin{tikzpicture}
[bend angle=30, scale=.9, 
  pre/.style={<-,shorten
    <=0pt,>=stealth,>={Latex[width=1mm,length=1mm]},thick}, post/.style={->,shorten
    >=0,>=stealth,>={Latex[width=1mm,length=1mm]},thick}, rpost/.style={=>,shorten
    >=0,>=stealth,>={Latex[width=1mm,length=1mm]},thick}, place/.style={circle, draw=black,
    thick,minimum size=5mm}, transition/.style={rectangle, draw=black,
    thick,minimum size=5mm}, revtransition/.style={rectangle, draw=red,
    thick,fill=red!50,minimum size=5mm}, invplace/.style={circle,
    draw=black!0,thick}]
\node[place,tokens=1] (b1) at (0,3) [label=left:\annone] {};
\node[place] (b2) at (3,3) [label=below:\anntwo] {};
\node[place] (b3) at (4.5,2.25) [label=right:\annthree] {};
\node[place,tokens=1] (b4) at (1.5,1) [label=below:\annfour] {};
\node[place] (b5) at (4.5,1)  [label=below:\annfive] {};
\node[transition] (e1) at (1.5,2.4) {$e_1$}
edge[pre] (b1)
edge[post] (b2)
edge[post, bend right] (b3);
\node[transition] (e2) at (3,1) {$e_2$}
edge[pre] (b4)
edge[pre] (b3)
edge[post] (b5);
\node[revtransition] (re1) at (1.5,4) {$\underline{e}_1$}
edge[pre, bend left] (b2)
edge[pre, bend left] (b3)
edge[post, bend right] (b1);
\end{tikzpicture}}

 \end{center}
 %
\noindent
The constructions we have proposed are somehow the more adherent to what is usually done, based on the 
interpretation that \emph{causality} implies that the event causing some other event somehow
produces something that it is used by the latter. 
This is not the only interpretation of what causality could mean. 
In fact, causality is often confused with the observation that two causally related events appear
ordered in the same way in each possible execution, and when we talk about ordered execution, 
it should be stressed that this can be achieved in several ways, for instance using \emph{inhibitor}
arcs. 
Consider the net $C$:
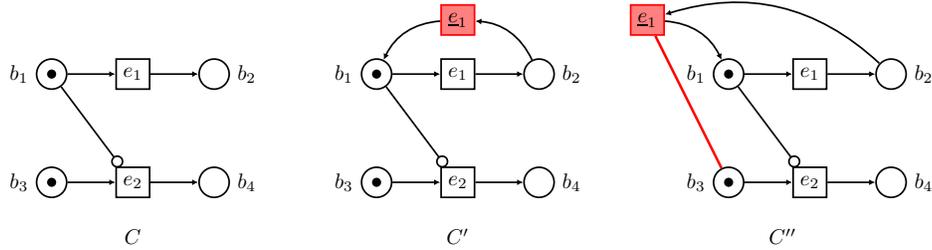
\begin{center}
\newcommand{\annone}{$(e_1,\ast)$}
\newcommand{\anntwo}{$(\ast,e_1)$}
\newcommand{\annthree}{$(e_1,e_2,<)$}
\newcommand{\annfour}{$(e_2,\ast)$}
\newcommand{\annfive}{$(\ast,e_2)$}
\scalebox{0.8}{
\begin{tikzpicture}
[bend angle=30, scale=.9, 
  pre/.style={<-,shorten
    <=0pt,>=stealth,>={Latex[width=1mm,length=1mm]},thick}, post/.style={->,shorten
    >=0,>=stealth,>={Latex[width=1mm,length=1mm]},thick}, prei/.style={o-,shorten
    <=0pt,>=stealth,>={Latex[width=1mm,length=1mm]},thick}, read/.style={-,shorten
    <=0pt,,>=stealth,>={Latex[width=1mm,length=1mm]},draw=red,very thick}, place/.style={circle, draw=black,
    thick,minimum size=5mm}, transition/.style={rectangle, draw=black,
    thick,minimum size=5mm}, revtransition/.style={rectangle, draw=red,
    thick,fill=red!50,minimum size=5mm}, invplace/.style={circle,
    draw=black!0,thick}]
\node[place,tokens=1] (b1) at (0,3) [label=left:$b_1$] {};
\node[place] (b2) at (3,3) [label=right:$b_2$] {};
\node[place,tokens=1] (b4) at (0,1) [label=left:$b_3$] {};
\node[place] (b5) at (3,1)  [label=right:$b_4$] {};
\node[transition] (e1) at (1.5,3) {$e_1$}
edge[pre] (b1)
edge[post] (b2);
\node[transition] (e2) at (1.5,1) {$e_2$}
edge[pre] (b4)
edge[prei] (b1)
edge[post] (b5);
\node (name) at (1.5,0) {$C$};
\node[place,tokens=1] (b11) at (6,3) [label=left:$b_1$] {};
\node[place] (b21) at (9,3) [label=right:$b_2$] {};
\node[place,tokens=1] (b41) at (6,1) [label=left:$b_3$] {};
\node[place] (b51) at (9,1)  [label=right:$b_4$] {};
\node[transition] (e11) at (7.5,3) {$e_1$}
edge[pre] (b11)
edge[post] (b21);
\node[revtransition] (re11) at (7.5,4) {$\underline{e}_1$}
edge[pre, bend left] (b21)
edge[post, bend right] (b11);
\node[transition] (e21) at (7.5,1) {$e_2$}
edge[pre] (b41)
edge[prei] (b11)
edge[post] (b51);
\node (name) at (7.5,0) {$C'$};
\node[place,tokens=1] (b12) at (12.5,3) [label=left:$b_1$] {};
\node[place] (b22) at (15.5,3) [label=right:$b_2$] {};
\node[place,tokens=1] (b42) at (12.5,1) [label=left:$b_3$] {};
\node[place] (b52) at (15.5,1)  [label=right:$b_4$] {};
\node[transition] (e12) at (14,3) {$e_1$}
edge[pre] (b12)
edge[post] (b22);
\node[revtransition] (re12) at (11,4) {$\underline{e}_1$}
edge[pre, bend left] (b22)
edge[read] (b42)
edge[post, bend left] (b12);
\node[transition] (e22) at (14,1) {$e_2$}
edge[pre] (b42)
edge[prei] (b12)
edge[post] (b52);
\node (name) at (13.5,0) {$C''$};
\end{tikzpicture}}

\end{center}
the event $e_2$ can be executed only after the event $e_1$ has been executed. However, $e_1$ 
does not produce a token (resource) that must be used by $e_2$. If we simply make
the event $e_1$ reversible but do nothing to prevent reversing of $e_1$ before $e_2$ is reversed,
then we would obtain the net $C'$. We could do better in $C''$ where we model the prevention using the
so-called \emph{read arcs} \cite{MR:CN}. Hence, using the inhibitor or read arcs seem feasible way
forwards to capture more precise the new relations of  r\pes es, including prevention. 
A similar approach has been already pursued in
\cite{CP:soap17} to model the so called modifiers that are able to change the causality pattern of
an event.
This suggests that, for arbitrary r\pes es, we need to find relations different from the flow 
relation to capture faithfully (forward and reverse) causal and prevention dependencies.
This will be the subject of future research.


\newpage
\bibliography{biblio}

\newpage
\appendix
\section{Omitted Proofs}\label{sec:proof}

\begingroup
\def\thetheorem{\ref{pr:reachable-markings-of-rcn}}

\begin{proposition}
 Let $R = \langle B, E, \changed{E^r}{\Er}, F, \mathsf{c}\rangle$ be a \rcn.
 Then $\reachMark{R} = \reachMark{C_{\overline{E}}}$.
\end{proposition}
\begin{proof}
 One direction is trivial, namely $\reachMark{C_{\overline{E}}}\subseteq\reachMark{R}$. For the
 other direction, we first observe that $\neg (\mathsf{c}\trans{e})$ holds for all $e\in \changed{E^r}{\Er}$. This is because $C_{\overline{E}}$ is an occurrence net, and  
 this implies that $\forall b\in \mathsf{c}$. $\pre{b}$ is either the $\emptyset$
 or it contains elements in $E^r$, and $\forall e\in \changed{E^r}{\Er}$. $\pre{e}\cap\post{b} = \emptyset$.
 Now we show that if an event $\changed{e^r\in E^r}{\er \in \Er}$ is executed then the
 corresponding event $h(\changed{e^r}{\er})$ has been executed before. 
 W.l.o.g. we assume that all the events executed before $\changed{e^r}{\er}$ are the events in
 $E\setminus \changed{E^r}{\Er}$. 
 Consider the \fs\ $\sigma\trans{\changed{e^r}{\er}}\sigma'$, then 
 we have $\lead{\sigma}\trans{\changed{e^r}{\er}}$, which means that
 $\pre{\changed{e^r}{\er}}\subseteq \lead{\sigma}$, but the conditions 
 $\pre{\changed{e^r}{\er}}$ have been produced by the execution of a unique event, namely
 $h(\changed{e^r}{\er})$. Now we prove that $\sigma\trans{\changed{e^r}{\er}}m$ can be reached without executing
 both $\changed{e^r}{\er}$ and $h(\changed{e^r}{\er})$. Consider the marking $\lead{\sigma}$, as $\sigma\trans{\changed{e^r}{\er}}$
 we know that $\post{h(\changed{e^r}{\er})} \subseteq \lead{\sigma}$. 
 Now $\sigma$ can be rewritten as $\sigma''\trans{h(\changed{e^r}{\er})}\sigma'''$ and
 $h(\changed{e^r}{\er})$ is concurrent with all the events in $\sigma'''$, which means 
 that $\sigma$ can be rewritten as $\hat{\sigma}\trans{h(\changed{e^r}{\er})}\lead{\sigma}$.
 Now we have $m = \lead{\hat{\sigma}}$ which implies that each
 reachable marking can be reached executing the events in 
 $E\setminus \changed{E^r}{\Er}$ only, hence $\reachMark{R}\subseteq\reachMark{C_{\overline{E}}}$. 
\end{proof}

\begingroup
\def\thetheorem{\ref{prop:rev_net}}

\begin{proposition}
 Let $C = \langle B, E, F, \mathsf{c}\rangle$ be an occurrence net and let
 $\mathsf{R}(C) = \langle B, \hat{E}, \changed{E'}{\Er}\times\setenum{\re}, \hat{F}, \mathsf{c}\rangle$ be the net 
 defined in Definition \ref{de:constructing-a-reversible-causal-net} with respect to 
 $\changed{E'}{\Er} \subseteq E$. Then $\mathsf{R}(C)$ is a reversible \changed{occurrence}{causal} net with respect to
 $\changed{E'}{\Er}\times\setenum{r}$.
\end{proposition}
\begin{proof}
 We just have to prove that $\mathsf{R}(C)$ is a safe net; the other conditions are satisfied
 by construction. 
 First we observe that if $b\not\in\mathsf{c}$ and 
 $\pre{b}$ is not a singleton in $\mathsf{R}(C)$ then
 $\pre{b}$ contains at most one event of the form $(e,\fe)$, and it contains 
 at least one of the form $(e',\re)$, and these are originated by the events in
 $\post{b}$ in $C$.
 In the case $b\in\mathsf{c}$ and $\pre{b}$ is not empty, then again 
 $\pre{b}$ contains only events of the form $(e',\re)$, and these are originated by the 
 events in $\post{b}$ in $C$.
 Assume it is not, and assume that $b\in B$ is the condition which receives a token when it 
 is already marked. As $C$ is an occurrence net, if the condition is marked then the event $e\in E$
 such that $b\in\post{e}$ has been executed and none of the events $e'\in E$ such that 
 $e'\in\post{b}$ (if any) have yet been executed. Thus in $\mathsf{R}(C)$ the
 event $(e,\fe)$ has been executed and none of the events $(e',\fe)\in\post{b}$ has been executed
 yet. But to be marked again an event of the form $(e'',\re) \in \pre{b}$ should have occurred,
 but this is impossible as none of the events $(e',\fe)\in\post{b}$ have been executed, and among 
 these also $(e'',\fe)$, contradicting the fact that the condition $b$ receives another token. 
\end{proof}

\begingroup
\def\thetheorem{\ref{pr:rce-to-rpes}}

\begin{proposition}
 Let $R = \langle B, E, \changed{E^r}{\Er}, F, \mathsf{c}\rangle$ be a reversible \changed{occurrence}{causal} net with respect to $\changed{E^r}{\Er}$, 
 then $\mathcal{C}_{r}(R) = (E', \changed{E''}{\Er'}, <, \#, \prec, \triangleright)$ is its associated r\pes, where
 \begin{itemize}
  \item $E' = E\setminus \changed{E^r}{\Er}$ and $\changed{E''}{\Er'} = h(\changed{E^r}{\Er})$,
  \item $<$ is the transitive closure of the relation $<'$ defined on the occurrence net
        $C_{\overline{E}}$ as $e <' e'$ whenever $\post{e} \cap \pre{e'}\neq \emptyset$,
  \item $\#$ is the conflict relation defined on the occurrence net $C_{\overline{E}}$,
  \item $e\ \triangleright\ \underline{e'}$ whenever $e\in \future{e'}$, and 
  \item $e\ \prec\ \underline{e'}$ whenever $e = e'$, and
\item $\ll = <$.        
 \end{itemize}
\end{proposition}

\begin{proof}
 First of all it is quite clear that $(E', <, \#)$ is a p\pes\ (if we close $<$ reflexively we get
 indeed a \pes), as it is obtained by $C_{\overline{E}}$.
 The relation $\prec \subseteq E'\times \changed{E''}{\Er'}$ satisfies the requirement that $e \prec \underline{e}$ 
 and that $\setcomp{e'}{e'\prec\underline{e}}$ is finite for each $e\in \changed{E''}{\Er'}$ as it contains 
 just $e$.
 If $e \prec \underline{e}$ then not $e \triangleright \underline{e}$ as $e\not\in\future{e}$.
 The sustained causation relation $\ll$ coincides with the relation $<$ hence the conflict relation
 is inherited along this relation. Furthermore, for $e\in \changed{E''}{\Er'}$, if $e < e'$ for some $e'$, 
 then we have that $e' \triangleright e$, as required.
 We can then conclude that $\mathcal{C}_{r}(R)$ is  a r\pes.
\end{proof}

\begingroup
\def\thetheorem{\ref{prp:rpes_respect}}

\begin{proposition} Let $R = \langle B,  E, \changed{E^r}{\Er}, F, \mathsf{c}\rangle$ be a reversible \changed{occurrence}{causal} net with respect to $E^r$ and 
 $\mathcal{C}_{r}(R) = (E', \changed{E''}{\Er'}, <, \#, \prec, \triangleright)$ be the associated r\pes. 
 Then $\mathcal{C}_{r}(R)$ is a 
 causal r\pes.
\end{proposition}
\begin{proof}
 Easy inspection of the construction in Proposition \ref{pr:rce-to-rpes}.
 The sustained causality $\ll$ clearly coincides with $<$.
 If $e \prec \underline{e'}$ then $e' = e$ and by construction if $e \triangleright \underline{e'}$ 
 then $e' < e$ as $e\in \future{e'}$.
\end{proof}

\begingroup
  \def\thetheorem{\ref{th:rnctorpes-conf-correspond}}
  
  \begin{theorem}
   Let $R = \langle B, E, \changed{E^r}{\Er}, F, \mathsf{c}\rangle$ be a reversible causal net with respect to $\changed{E^r}{\Er}$ and
   $\mathcal{C}_{r}(R) = (E', \changed{E''}{\Er'}, <, \#, \prec, \triangleright)$ be the associated r\pes.
   Then $X \subseteq E'$ is a configuration of $R$ iff $X$ is a configuration of $\mathcal{C}_{r}(R)$.
  \end{theorem}
  \begin{proof}
   As $\mathcal{C}_{r}(R)$ is a cause-respecting and causal r\pes\ we have that each configuration
   is forward reachable, and the forward reachable configurations are precisely those conflict-free
   and left-closed of the p\pes\ $\mathcal{C}_{r}(R) = (E', <, \#)$, which correspond to the
   configurations of the occurrence net $R_{\overline{E}}$.
  \end{proof}

\begingroup
\def\thetheorem{\ref{prop:mixed_step}}
\begin{proposition}
 Let $R = \langle B, E, \changed{E^r}{\Er}, F, \mathsf{c}\rangle$ be a reversible \changed{occurrence}{causal} net and let 
 $\mathcal{C}_{r}(R) = (E', \changed{E''}{\Er'}, <, \#, \prec, \triangleright)$ be the associated r\pes. 
 Let $X\in\Conf{R}{\rcn}$ and let $A\subseteq E$ be such that 
 $\mathsf{mark}(X)\trans{A}$. Then $\hat{A}\cup \underline{B}$ is enabled at $X$ in
 $\mathcal{C}_{r}(R)$, where $\hat{A} = \setcomp{e\in A}{e\not\in \changed{E^r}{\Er}}$ and
 $\underline{B} = \setcomp{e\in A}{e\in \changed{E^r}{\Er}}$.
\end{proposition}
\begin{proof}
 By Theorem \ref{th:rnctorpes-conf-correspond} we know that $X\in \Conf{\mathcal{C}_{r}(R)}{r\pes}$.
 We have to check that $\hat{A}\cup \underline{B}$ is enabled at $X$. 
 As $\mathsf{mark}(X)\trans{A}$ we know that $\pre{A}\subseteq \mathsf{mark}(X)$,
 hence $A\cap X$ should be equal to $\emptyset$. Furthermore for any
 $e\in A\cap \changed{E^r}{\Er}$, as  $\mathsf{mark}(X)\trans{\setenum{e}}$, we have
 that $h(e)\in X$ (otherwise the conditions enabling $e$ would not have been produced), 
 and then we have that $B = \setcomp{h(e)}{e\in \underline{B}}\subseteq X$. 
 Finally, as $\mathsf{mark}(X)\trans{A}$, we have that
 $\CF{X\cup\hat{A}}$ holds. 
 Consider now $e\in \hat{A}$, and $e' < e$. Clearly $e'\in X\setminus B$.
 Assume the contrary, then $e'\in B$ and there exists an $\underline{e}'\in A\cap \changed{E^r}{\Er}$
 such that $h(\underline{e}') = e'$, but then we have that $\neg \mathsf{mark}(X)\trans{A}$.
 Consider now $e\in B$ (which means that $\underline{e}\in A\cap \changed{E^r}{\Er}$) and $e'\prec \underline{e}$.
 As $\mathcal{C}_{r}(R)$ is a causal r\pes, we know that $e' = e$ and 
 $e\in X\setminus (B\setminus\setenum{e})$.
 Take now $e\in B$ and $e'\triangleright\underline{e}$. This means that $e'\in\future{e}$ which
 implies that $e\not\in X$, and also that $e\not\in \hat{A}$.
 By Definition \ref{de:rpes-conf-enab} we can conclude that $\hat{A}\cup \underline{B}$ is enabled at $X$.
 Finally we observe that $\mathsf{mark}(Y) = \mathsf{c}'$ where 
 $\mathsf{mark}(X)\trans{A}\mathsf{c}'$ and $X \stackrel{\hat{A}\cup\underline{B}}{\longrightarrow} Y$.
\end{proof}

\begingroup
\def\thetheorem{\ref{th:rnctorpes-conf-correspond}}

\begingroup
\def\thetheorem{\ref{prp:rpes_to_cnet}}

\begin{proposition}
 Let $\mathsf{P} = (E, \changed{E'}{\Er}, <, \#, \prec, \triangleright)$ be a causal r\pes\ and
 let $<^{+}$ be the transitive closure of $<$. Then $\#$ is inherited along  $<^{+}$,
 \emph{i.e.} $e\ \#\ e' <^{+} e''\ \Rightarrow\ e\ \#\ e''$.
\end{proposition}
\begin{proof}
 In general we have that, given a r\pes, $(E, \ll, \#)$ is a \pes. But in a causal r\pes\ we have
 that $\ll$ is indeed the transitive closure of $<$.
\end{proof}

\begingroup
\def\thetheorem{\ref{prp:rpes_to_cnet}}

\begin{proposition}
 Let $\mathsf{P} = (E, \changed{E'}{\Er}, <, \#, \prec, \triangleright)$ be a causal r\pes. Then  
 $\mathcal{E}_{r}(\mathsf{P}) = \langle B, \hat{E}, E'\times\setenum{\re}, F, \mathsf{c}\rangle$ as defined
 in Definition \ref{de:rpestorcn} is a reversible \changed{occurrence}{causal} net with respect to $E'\times\setenum{\re}$.
\end{proposition} 
\begin{proof}
  By construction $\mathcal{E}_{r}(\mathsf{P})_{E\times\setenum{\fe}}$ is a occurrence net.
  The other requirements can be easily checked. For each $(e,\re)$ there exists a
  unique event $(e,\fe)$, and if two events share the same preset and postset they are
  clearly the same event. Each condition $b\in B$ is clearly related to an event
  in $E\times\setenum{\fe}$ hence in $\hat{E}\setminus (E'\times\setenum{\re})$. 
\end{proof}

\begingroup
\def\thetheorem{\ref{th:cccrpestorcn}}
\begin{theorem}
 Let $\mathsf{P}$ be a causal 
 r\pes. Then  
 $X'$ is a configuration of $\mathcal{E}_{r}(\mathsf{P})$ iff $X$ 
 is a configuration of $\mathsf{P}$, where $X' = \setcomp{(e,\fe)}{e\in X}$.
\end {theorem}
\begin{proof}
 Let $\mathsf{P} = (E, \changed{E'}{\Er}, <, \#, \prec, \triangleright)$. 
 Consider $X\in \Conf{\mathsf{P}}{r\pes}$. As $\mathsf{P}$ is a cause-respecting and
 causal r\pes\ we have that $X$ is forward reachable, hence $X$ is a configuration
 of the p\pes\ $(E, <, \#)$, which we denote with $P$, and then $X' = \setcomp{(e,\fe)}{e\in X}$ 
 is a configuration also of the occurrence net associated to this event structure as,
 by Proposition \ref{pr:ppes-prop}, we have that $\Conf{P}{p\pes} = \Conf{\mathsf{hc}(P)}{\pes}$. 
 For the vice versa it is enough to observe that, up to renaming of events, 
 $\mathcal{C}_{r}(\mathcal{E}_{r}(\mathsf{P}))$ is indeed $\mathsf{P}$.
\end{proof}

\begingroup
\def\thetheorem{\ref{pr:C-is-functor}}
\begin{proposition}
 $\mathcal{C}_{r} : \mathbf{rOcc} \to \mathbf{cRPES}$ acting on objects as in Proposition~\ref{pr:rce-to-rpes}
 and on morphisms $(\beta,\eta) : R_0 \to R_1$ as $\eta$ restricted to the events that are not
 a reversing ones, is a functor.
\end{proposition}
\begin{proof}
 For the objects part we have Proposition~\ref{pr:rce-to-rpes}.
 For the part on morphisms is enough to observe that the requirements to fulfill are the last two
 of Definition~\ref{de:rpes-morph}. For the first one is enough to observe that
 the prevention relation is induced by the causality relation, and 
 $\eta(e_0) \triangleright_1 \eta(e_0')$ then $\eta{e_0'}\in\hist{f(e_0)}$ which means
 that $e_0'\in\hist{e_0}$ as $\eta$ is defined on both, and then $e_0 \triangleright_0 e_0'$.
 For the last point we have that a reversing event $\underline{e}$ is preceded in the
 reversible causal net by $e$ itself alone, hence we have that $e \prec_0 \underline{e}$.
 Now reversing events are preserved by reversible causal nets, and if $\eta{e}$ is defined
 we have also that the reversing event is defined which implies that 
 $\eta{e} \prec_1 \underline{\eta{e}}$. The thesis follows. 
\end{proof}

\begingroup
\def\thetheorem{\ref{pr:E-is-functor}}
\begin{proposition}
 $\mathcal{E}_r : \mathbf{cRPES} \to \mathbf{rOcc}$ acting on objects as in Definition~\ref{de:rpestorcn}
 and on morphisms as stipulated in for occurrence net, requiring that reversing events are
 preserved, is a functor.
\end{proposition}
\begin{proof}
 The only condition to check is that reversing events are preserved, but this trivially true as
 causal r\pes\ morphisms do this.
\end{proof}

\begingroup
\def\thetheorem{\ref{th:coreflection}}
\begin{theorem}
 $\mathcal{E}_{r}$ and $\mathcal{C}_{r}$ form a coreflection, where $\mathcal{C}_{r}$ is the right
adjoint and $\mathcal{E}_{r}$ the left adjoint.
\end{theorem}
\begin{proof}
 We observe that $\mathcal{C}_{r}(\mathcal{E}_{r}(P)) = P$ and the identity mapping 
 $1_{P} : P \to \mathcal{C}_{r}(\mathcal{E}_{r}(P))$ is free over $P$ with respect to $\mathcal{C}_{r}$,
 \emph{i.e.} given any other reversible causal net $R$ and any morphism $f : P \to \mathcal{C}_{r}(R)$,
 then there exists a unique reversible causal net morphism from $\mathcal{C}_{r}(R)$ to $R$.
 But this is a consequence of the freeness of the coreflection between $\mathcal{E}$ and $\mathcal{C}$,
 as the unique mapping act on the reversing events as prescribed by the causal reversible prime 
 event structure morphism. 
\end{proof}

\end{document}